\documentclass[12pt, onecolumn, draftclsnofoot, journal]{IEEEtran}
\usepackage{color}
\usepackage{graphicx}
\usepackage{amsmath,bm}
\usepackage{amssymb}
\usepackage{amsthm}
\usepackage{cite}
\usepackage{colortbl}
\usepackage[export]{adjustbox}
\usepackage{enumitem}
\usepackage{float}
\usepackage{epstopdf}
\usepackage[export]{adjustbox} 
\usepackage{ulem}
\newtheorem{theorem}{Theorem}

\usepackage{caption}
\usepackage{subfig} 
\usepackage{lipsum}
\usepackage{stfloats}
\usepackage{graphicx}
\IEEEoverridecommandlockouts

\usepackage{cite}

\newcommand{\lb}{\left(}
\newcommand{\rb}{\right)}

\newcommand{\lcb}{\left\{}
\newcommand{\rcb}{\right\}}

\newcommand{\nbold}{\mathbf{n}}

\usepackage{mathtools, cuted}
\usepackage{lipsum, color}

\usepackage{cite}

\graphicspath{{Pictures/}}
\usepackage{amsmath,amssymb,amsthm}
\newtheorem{definition}{Definition}
\usepackage{optidef}
	
\begin{document}	
\title{Performance Analysis of a MIMO System with Bursty Traffic in the presence of Energy Harvesting Jammer}
\author{Sujatha~Allipuram,~\IEEEmembership{Student~Member,~IEEE,}
	Parthajit~Mohapatra,~\IEEEmembership{Member,~IEEE,} Nikolaos Pappas,~\IEEEmembership{Senior Member,~IEEE,} Shabnam Parmar
	and Saswat~Chakrabarti,~\IEEEmembership{Senior Member,~IEEE}
	\thanks{This work was supported in part by the joint research project funded by the Swedish Research Council (VR Sweden) and Department of Science and Technology (DST), India. A part of this work has appeared in \cite{allipuram-ncc-2020}.}
	\thanks{Sujatha Allipuram and Saswat Chakrabarti are with the G.S Sanyal School of Telecommunication, Indian Institute of Technology Kharagpur,
		WB, India-721302  (e-mail: asujatha@iitkgp.ac.in, saswat@ece.iitkgp.ac.in).}
	\thanks{Parthajit Mohapatra is with the Department of Electrical Engineering, Indian Institute of Technology Tirupati, AP, India-517506   (e-mail:parthajit@iittp.ac.in).}
    \thanks{Nikolaos Pappas is with the Department of Science and Technology, Link\"{o}ping University, SE-60174 Norrk\"{o}ping, Sweden (e-mail:nikolaos.pappas@liu.se).}
	\thanks{Shabnam Parmar is with Intel Technology India Pvt. Ltd., Bangalore, India-560017 (e-mail:shabnam.parmar@intel.com).}
}
\maketitle
\vspace{-20pt}
\begin{abstract}
	This paper explores the role of multiple antennas in mitigating jamming attacks for the Rayleigh fading environment with exogenous random traffic arrival. The jammer is assumed to have energy harvesting ability where energy arrives according to Bernoulli process. The outage probabilities are derived with different assumptions on the number of antennas at the transmitter and receiver. The outage probability for the Alamouti space-time code is also derived. The work characterizes the average service rate for different antenna configurations taking into account of random arrival of data and energy at the transmitter and jammer, respectively. In many practical applications, latency and timely updates are of importance, thus, delay and Average Age of Information (AAoI) are the meaningful metrics to be considered. The work characterizes these metrics under jamming attack. The impact of finite and infinite energy battery size at the jammer on various performance metrics is also explored. Two optimization problems are considered to explore the interplay between AAoI and delay under jamming attack. Furthermore, our results show that Alamouti code can significantly improve the performance of the system even under jamming attack, with less power budget. The paper also demonstrates how the developed results can be useful for multiuser scenarios.
	\end{abstract}
\IEEEpeerreviewmaketitle
\vspace{-10pt}
\section{Introduction}	
Internet of Things (IoT) networks are vulnerable to attacks due to the broadcast and superposition nature of the wireless medium. Jamming is a common form of denial of service (DoS) attack which can  significantly impact the IoT network's performance. In many IoT scenarios such as status update systems, users may not always have data to send; rather, data arrival at the users are random. It is also required to ensure  timely delivery of data at the destination within a given period of time due to requirements in terms of delay, as well as it is required to keep the information at the destination as fresh as possible. For such scenarios, along with stable throughput, it is required to consider delay and age of information (AoI), which are meaningful metrics to take account of latency in the communication.   In recent years, AoI has been used to capture the freshness of information \cite{AoINOW}. One of the concern arises is to ensure high throughput, low delay, and low AoI simultaneously under jamming attack, when the user also has energy constraint. The prospect of multiple antennas system has  been explored for green communication \cite{cui-jsac-2004, nguyen-tgcn-2018}. \textit{The role of multiple antennas in mitigating the jamming attack and enhancing the system performance when timely updates are important is not well explored in the existing literature and the primary focus of this work.}

In this paper, a point-to-point MIMO system in the presence of a jammer is considered, where the transmitter is equipped with a queue to store the incoming traffic. It is assumed that jammer has energy harvesting capability \cite{8068986,8768226,9170580,7294641}.  In many practical scenarios, it may not be possible for the jammer to have a constant source of energy supply due to the hostile nature of the environment, inaccessibility to the location, or absence of any external supply of energy such as a power grid.  However, it is possible to deploy jammer in such environments due to the advancements made in energy harvesting. When the attacker has energy harvesting ability, the deployment of such jammer will be easier and it can make them autonomous. When it is required to guarantee low latency along with reliability, even a random jammer can cause significant harm to the performance of the system, and hence, it is important to understand the impact of jamming on the reliability and latency of the system. The considered model helps to capture the impact of various parameters of the attacker such as the capacity of the battery, jamming probability, and jamming power on the performance of the system. From the jammer's perspective, the considered model can help to explore how large should be the battery size and energy harvesting ability to perform jamming effectively. As a special case, the considered model reduces to the case of a constant jammer.

This work develops a cross-layer framework for the point-to-point MIMO system in the presence of a jammer. The cross-layer framework captures the random arrival of data at the transmitter using network-level metrics such as stability of the queue and reliability of the underlying channel model through outage probability. The outage probability takes into account of Signal to Noise Ratio (SNR) or Signal to Jamming and Noise power Ratio (SJNR) and fading phenomena of the wireless channel between the various nodes. Such frameworks have been used to study different communication models in the existing literature \cite{lu-TVT-2016, wang-tvt-2017}. The work also demonstrates how the developed results in the paper can be useful to characterize stability region for multi-user scenarios. 

The paper first obtains the outage probabilities for different MIMO configurations under jamming attack. Then, the service probability at the transmitter is characterized  using the outage probability obtained for different antenna configurations at the legitimate nodes. The service probability is used to characterize the average delay per packet, and the average AoI (AAoI) of various antenna configurations under random arrival of data at the legitimate transmitter. These metrics are characterized for two scenarios where the jammer can have a battery of \textit{finite capacity} or \textit{unlimited capacity}. The role of space-time coding on performance is also investigated. The derived results can also be used to minimize AAoI or average delay or maximize the average service rate for optimal data arrival rate. In the later part of the paper, it is shown how the characterization of the stable throughput of the point-to-point system helps to characterize the stability region for the 2-user SIMO broadcast channel (BC) in the presence of an energy harvesting jammer.
\vspace{-10pt}
\subsection{Related works}
The conventional security mechanisms generally do not provide protection or detection of jamming attacks in wireless networks. The jammer can adopt different strategies to carry the attack, and the various jamming mitigation techniques can be classified into the following categories: channel hopping, coding protection, rate-adaption, and MIMO-based jamming mitigation \cite{xu2005feasibility}. Channel hopping and spread spectrum are two commonly used techniques to mitigate the jamming attack. One of the disadvantages of the channel hopping-based method is the requirement of preshared channel assignment.   Interference alignment which is used generally to mitigate interference has been used for MISO broadcast channels based on topology for mitigation of jamming attack \cite{amuru-tit-2014}. To understand the fundamental limits on the  performance under jamming attacks, tools and techniques from information theory have been used for various communication models \cite{ kashyap-tit-2004, amuru-tit-2014}.




The attacker can also have the ability to eavesdropping as well as jamming \cite{8968374, 7470273, ryu2016transmission}. The problem of secure communication over a correlated fading channel is considered  in  \cite{7470273} for multi-user multi-cell massive MIMO in the presence of an active eavesdropper equipped with multiple antennas. The work shows that transmit antenna correlation diversity between the nodes can be exploited to improve the performance even under pilot contamination attacks. The work in \cite{ryu2016transmission} develops a transmission strategy that involves the determination of secure transmission rate and beam-forming design based on complete or partial knowledge of the jamming channel in case of MISO system. A friendly jammer can also be used to enhance the secrecy performance by sending a jamming signal to degrade the SNR at the eavesdropper  \cite{8121977, 8351952}. In general, there is complete trust between the jammer and legitimate nodes in such scenarios.




Legitimate users always try to achieve the desired performance and on the other hand, attacker attempt to cause maximum harm to the communication. To capture this conflict arising between jammer and legitimate users, game theory has been used to study various models under jamming attacks \cite{kashyap-tit-2004, ulukus-mcc-2005, garnaev-twc-2016, mukherjee-tsp-2013, garnev-wifs-2014}. In \cite{kashyap-tit-2004}, using a zero-sum mutual information game it is shown that knowing the input signal at the jammer is not useful for the point-to-point MIMO Rayleigh fading channel.  The work in \cite{ulukus-mcc-2005} considers a non-cooperative zero-sum game for 2-user multiple access channels in the presence of a jammer under different assumptions on channel characteristics. It is shown that when jammer does not know users' signals, the solution for the game exists. Another important problem is the allocation of resources in the presence of the attacker. The work in \cite{garnev-wifs-2014} uses the Bayesian game using an $\alpha$-fairness utility for resource allocation under an unknown jamming attack. The tool from game theory has also been used in the case of active eavesdropper where the attacker can either jam, eavesdrop on the ongoing communication, or both \cite{mukherjee-tsp-2013, garnaev-twc-2016}.


The works discussed so far assume that users always have data to send. However, in many practical scenarios, data arrival at the users can be random. The game-theoretic framework has been used to explore the role of random arrival of data in mitigating jamming attacks in \cite{sagduyu-isit-2010}. The framework of the non-cooperative game has been used in \cite{5518798} to determine the impact of jamming over a collision channel.  For many applications, where a delay and timely updates are of importance, throughput alone cannot capture these attributes. For such scenarios, delay and AoI are relevant metrics to be considered, and these metrics can give new insights into system performance. AoI so far has been studied under different queuing disciplines in simple point-to-point systems, in more elaborated systems with multiple access, or scheduled access  \cite{AoINOW,sun2019age}. The work in \cite{chen2020multiuser} studied the user scheduling problem in a MIMO status update system, where multiple users with a single antenna aim to send their latest updates to an information-fusion access point equipped with multiple antennas via a shared wireless channel. The authors derived an optimal scheduling policy that can minimize the AoI over the networks. Furthermore, the effect of jamming on AoI has been studied in \cite{SunAoIW2018, GarnaevAoIW2019}. Characterizing the stability region for multiuser scenarios is a challenging problem due to the interaction between the queues \cite{8320826, 8357571}. However, the impact of jamming on the achievable throughput, delay, and AoI in a system with random traffic is not well understood in fading scenarios when users are equipped with multiple antennas. This also brings an important question on how diversity (time, spatial, or both) can be exploited to improve the performance of the system in the presence of a jammer. This work primarily aims to answer these questions. 
\vspace{-15pt}
\subsection{Contributions}
The main contributions of the work are summarized below.
\begin{enumerate} \itemsep 0.2em
\item The outage probabilities are obtained for the Rayleigh fading environment in the presence of jammer under the different assumptions on the number of antennas at the transmitter and receiver. The outage probability for Alamouti coding scheme is also derived for the considered model. As the signal sent by the transmitter and attacker undergo Rayleigh fading, it is a non-trivial problem to determine the distribution corresponding to the signal to jamming and noise ratio (SJNR). To the best of the authors' knowledge, the role of space-time code in mitigating jamming attack has not been explored in the existing literature.
 
\item Characterization of stable throughput or stability region under random arrival of data in the presence of jammer is a non-trivial problem. The paper first aims to characterize the stable throughput under jamming attack for the point-to- point MIMO system. The average service rate is characterized for the considered system model using the outage probabilities obtained for different multiple antennas setup. To capture the energy harvesting ability of the jammer two scenarios are considered: \textit{battery with unlimited capacity} and \textit{battery with limited capacity}. The developed results take account of the random arrival of data at the transmitter and energy at the attacker. 

\item The work also characterizes the average delay performance and AoI  for various multi-antenna configurations in presence of the attacker. The developed results help to explore the interplay between delay and AAoI for different antenna configurations under jamming attack and developed results reconfirm the utility of the AAoI metric for latency-aware communication. Two optimization problems are considered where the objective is to minimize the AAoI with different delay requirements.

\item The work also characterizes the stability region of the 2-user SIMO broadcast channel  under jamming attack for the Rayleigh fading scenario using the results obtained for the point-to-point model. In this case, receiver decodes its intended message by treating other user's signal as noise.
\end{enumerate}
The obtained results in this work provide a unified view of the role of multiple antennas in improving the performance of the system in terms of stable throughput, delay, and AAoI under jamming attack. It is also found that Alamouti coding can achieve minimum delay and AoI among the different MIMO configurations in the setup with less power budget. The Alamouti coding scheme has the added advantage of not requiring channel state information at the transmitter. The developed results for the point to point MIMO system can act as a basic building block to characterize the stability region of multiuser scenarios, which is a challenging problem due to the interaction among the queues. 
\vspace{-3pt}
\section{System Model}
\vspace{-10pt}
\begin{figure}[h]
	\centering
	\includegraphics[trim={0cm 0cm 0cm 0cm},clip, height=1.8in, width=3.35in]{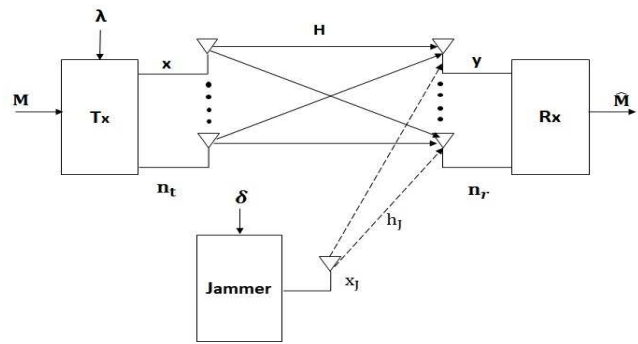}
	\setlength{\belowcaptionskip}{-6pt} %
	\caption{Point-to-point MIMO system in the presence of energy harvesting jammer.}\label{fig:figure1}
\end{figure}
This paper considers a point-to-point MIMO system in the presence of a jammer, where transmitter and receiver are equipped with $n_t$ and $n_r$ antennas, respectively. The transmitter has a queue, which stores the incoming packet, and it needs to be sent reliably to the receiver in the presence of the attacker. A pictorial description of the model is shown in Fig~\ref{fig:figure1}. 
\vspace{-10pt}
\subsection{Network layer model}
The packet arrival process at the queue is assumed to be independent and stationary with mean probability $ \lambda$, i.e., Bernoulli \cite{9233374}. It is also assumed that the queue at the transmitter has unlimited capacity. The average length of the queue is denoted by $\bar{Q}$ and serves the packets with average service rate $\mu$. It is assumed that the jammer does not have a constant source of power supply, but it has energy harvesting ability.  The energy arrival process at the jammer is modeled by a Bernoulli process where the rate of energy arrival is assumed to be $\delta$ 
 \cite{9365698, 7959595, 9086254,7487983, 7009995}. These chunks are stored in a battery (an energy buffer). It is assumed that if one chunk of energy is harvested at the jammer, it is enough to jam the communication  with fixed jamming power $P_J$.  In the paper, two scenarios are considered for the energy buffer: (a) \textit{Energy buffer with unlimited capacity (Section~IV-A)}; and (b) \textit{Energy buffer with limited capacity (Section~IV B)}.
 \vspace{-10pt}
 \subsection{Physical layer model}
The input-output relation for the considered system model in the presence of jammer is given as follows:
\begin{align}
	& \mathbf{y} = \mathbf{H} \mathbf{x} + \mathbf{h_J} x_J + \nbold, \label{eq:sysmodel1}
\vspace{-10pt}
\end{align}
where the channel between transmitter and receiver is denoted by $\mathbf{H} \in \mathcal{C}^{n_r \times n_t}$ and channel between the jammer and receiver is denoted by $\mathbf{{h}_J} \in \mathcal{C}^{n_r \times 1}$. The input $\mathbf{x}$ to the channel is drawn from a Gaussian codebook, i.e., $\mathbf{x} \sim \mathcal{CN}(0, P \mathbf{I})$, and P is the transmitter's power budget. The attacker sends the jamming signal $x_J$ drawn from Gaussian distribution with jamming power $P_J$, i.e.,  $x_J \sim \mathcal{CN}(0, P_J)$. Due to the limited computing ability and constrained in the energy of the attacker, it is assumed that the jammer has a single antenna. The noise at the receiver is modeled as AWGN, i.e., $\mathbf{n} \sim \mathcal{CN}(\mathbf{0}, \mathbf{I})$. In the case of MISO, $\mathbf{H} \in \mathcal{C}^{1 \times n_t}$, and $\mathbf{n} \sim \mathcal{CN}(0, 1)$. For SIMO, $\mathbf{H} \in \mathcal{C}^{n_r \times 1}$ and input is drawn from a Gaussian codebook, i.e., $\mathbf{x} \sim \mathcal{CN}(0, P)$. The channel between the various nodes undergoes an independent Rayleigh fading process. The attacker does not have CSI and it cannot access the queue state information of the transmitter.  The transmitter and receiver need to know the statistical knowledge of the CSI between the jammer and receiver. When the battery is not empty, the attacker jams with a probability $p_J$. This type of jammer is known as a random jammer in the literature and belongs to the category of proactive jammer~\cite{10.1504/IJAHUC.2014.066419}.
\vspace{-10pt}
\subsection{Stability of the queue and average service rate}
The stability of the queue is defined as follows \cite{szpankowski_1994}.
\begin{definition}
The queue is said to be stable if the following condition is satisfied $\displaystyle\lim_{t \to \infty} Pr[Q^t<x] = F(x)$ and $\displaystyle\lim_{x \to \infty}  F(x) = 1$, where $Q^t$ denotes the length of queue at the beginning of the time slot $t$. When the following condition is satisfied $\lim_{x \to \infty}\lim_{t \to \infty}\text{inf } Pr[Q^t<x]= 1$, the queue is called sub-stable. If a queue is stable, then it is also sub-stable. The queue is said to be unstable, if the queue is not sub-stable. 
\end{definition}
Hence, the queue is stable if the average data arrival rate ($\lambda$) is less than the average service rate ($\mu$), i.e., $\lambda < \mu$. This result holds under the assumption that the queue's arrival and service processes are strictly jointly stationary. For the system model considered in this paper, the average service rate is defined as follows:
\vspace{-10pt}
\begin{align}
& \mu = \hspace{1mm} (1-p_J)\hspace{1mm}(1-p_{WoJ}^{\text{out}}) + \hspace{1mm} p_J \hspace{1mm} (1-p_{J}^{\text{out}}), \label{eq:servicerate1}
\end{align}
where  $p_{J}^{\text{out}}$ and $p_{WoJ}^{\text{out}}$ are the probability of unsuccessful decoding of the packet at the receiver with and without jamming attack, respectively. To characterize the service rate, it is required to obtain these probabilities. In this work, the outage probability is used as a proxy for the probability of unsuccessful decoding of data at the receiver \cite{lu-TVT-2016, wang-tvt-2017}. To evaluate the service rate,  legitimate nodes need to know the jamming power $P_J$ and jamming probability $p_J$. In practice, estimating these parameters is a non-trivial problem and some of the works in this direction can be found in \cite{6552344, article, article2, 6179344, 9039615}. Tools from machine learning can also be useful in estimating these parameters \cite{8979256}. 
\vspace{-6pt}
\section{Outage Probability for different antenna configurations under Jamming}\label{sec:outage}
In this section, the outage probabilities are obtained for different multi-antenna configurations for the Rayleigh fading scenario under a jamming attack.  The work also derives the outage probability for the MIMO with Alamouti coding. The role of space-time diversity in mitigating jamming attacks is not well understood from the existing literature.  The developed results help to explore the role of multiple antennas in mitigating jamming attacks. The derivation of these expressions is non-trivial as it is required to obtain the distribution associated with signal to jamming and noise ratio (SJNR) for multiple antenna configurations, where the message signal and the jamming signal go through the Rayleigh fading channel. These results are stated in the following theorem.
\begin{theorem}\label{th:firsttheorem}
The outage probabilities $(p_J^{\text{out}})$ for different antenna configurations under jamming attack are provided below.
\begin{enumerate}
\item The outage probability for MISO $(1 \times n_t)$ is given by the following expression:
\begin{align}
 p_J^{\text{out}} &=  1 - \frac{\Gamma(n_t,\frac{(2^R-1)n_t}{P})}{\Gamma(n_t)}  + \frac{e^{\frac{1}{P_J}}\Gamma(n_t,\frac{(2^R-1)n_t}{P}+\frac{1}{P_J})}{\Gamma(n_t)(1+\frac{P}{n_t(2^R-1)P_J})^{n_{t}}}. \label{eq:th1eq1}
\end{align}
\item The outage probability for SIMO $(n_r \times 1)$ is given by the following expression
\begin{align}
 p_{J}^{\text{out}} & =   1-\frac{\Gamma(n_r,\frac{(2^R-1)}{P})}{\Gamma(n_r)} +  \frac{e^{\frac{1}{P_J}}\Gamma(n_r,\frac{(2^R-1)}{P}+\frac{1}{P_J})}{\Gamma(n_r)(1+\frac{P}{(2^R-1)P_J})^{n_{r}}}. \label{eq:th1eq2}
\end{align}
\item The outage probability for  MIMO $(n_r \times n_t)$ with Alamouti code is given by the following expression
\begin{align}
 p_{J}^{\text{out}} & = 1- \frac{\Gamma(N, \frac{(2^R-1)}{\beta})}{\Gamma(N)} +  \dfrac{e^{\frac{1}{P_J}} \Gamma\bigg(N,\big( \frac{2^R-1}{\beta} + \frac{1}{P_J}\big)\bigg)}{\beta^N  \Gamma(N) \big[ \frac{1}{\beta} + \frac{1}{(2^R-1)P_J}\big]^N}. \label{eq:th1eq3}
\end{align}
\end{enumerate}
where $R$ is the given target rate, $\beta \triangleq \frac{P}{n_t}$, $N \triangleq n_tn_r$, $\Gamma(s)$ and $\Gamma(s,x)$ are the gamma function and the upper incomplete gamma function defined as follows, respectively
\begin{align}
\Gamma(s)=\int\limits_0^\infty t^{s-1}e^{-t}dt \text{ and } \Gamma(s,x)=\int\limits_x^\infty t^{s-1}e^{-t}dt.
\end{align}
\end{theorem}
\begin{proof} The proofs for different multi-antenna scenarios are given below:\\
\textit{Case 1 (MISO $(1 \times n_t)$):} In this case, the instantaneous achievable rate is given by the following expression
	\begin{equation}
	\begin{aligned}
	R_i = \log_2 \lb 1+\frac{|\mathbf{h}|^2\frac{P}{n_t}}{1+|{h_J}|^2 P_J}  \rb, \label{eq:misoach1}
	\end{aligned}
	\end{equation} 
	where $|\mathbf{h}|^2$ and $|{h_J}|^2$ are chi-square and exponentially distributed, respectively. Due to the lack of CSI at the transmitter, power is equally divided among all the transmit antennas. The outage probability is determined as follows
	\begin{equation} 
	\begin{split}
	p_{J}^{\text{out}} & = Pr\{ R_i <R \},   \\
	&= Pr\lcb \log_2 \lb 1+\frac{|\mathbf{h}|^2\frac{P}{n_t}}{1+|{h_J}|^2 P_J}  \rb <R \rcb,\\
	&= Pr\lcb \frac{|\mathbf{h}|^2P}{n_t} - (2^R-1)|{h_J}|^2 P_J  <2^R-1 \rcb. \nonumber \\
	\end{split}
\end{equation}
where $Pr(.)$ is the probability. 
To simplify the notation, substitute  $|\mathbf{h}|^2$ as U i.e., $U \sim \chi^2_{2n_{t}} $. $|{h_J}|^2$ as W i.e., $W \sim e^{-w}$ in the above expression and the outage probability becomes
	\begin{equation} 
	\begin{split}
   p_{J}^{\text{out}} & = 1- Pr\bigg\{ W < \frac{U P/n_t-(2^R-1)}{(2^R-1)P_J} \bigg\},  \\
	&= 1-\int\limits_{u=(2^R-1)n_t/P}^{\infty} \int\limits_{w=0}^{\frac{uP/n_t-(2^R-1)}{(2^R-1)P_J}} f_{W}(w) f_{U}(u) \:  dw \:du,\nonumber \label{use_SuccPr}
\end{split}
\end{equation}
\begin{equation} 
\begin{split}
&= 1-\int\limits_{u=(2^R-1)n_t/P}^{\infty} \!\!\!\!\! f_{U}(u) \: \bigg [ \int\limits_{w=0}^{\frac{uP/n_t-(2^R-1)}{(2^R-1)P_J}} f_{W}(w) dw \bigg] du, \\
&= 1-\!\!\!\!\!\int\limits_{u=(2^R-1)n_t/P}^{\infty}\!\!\!\!\!\!\!\!\!f_{U}(u)\ du  
+ e^{\frac{1}{P_J}}\!\!\!\!\!\!\!\!\!\int\limits_{u=(2^R-1)n_t/P}^{\infty} \!\!\!\!\!\!\!\! f_{U}(u)e^{\big(\frac{-Pu}{n_t(2^R-1)P_J}\big)} du, \label{eq:misopout}
	\end{split}
\end{equation}
where $f_{U}(u)$ and $f_{W}(w)$ are the probability density functions of U and W, respectively.
Solving both the integrals, the outage probability reduces to following
\begin{equation}
	\begin{aligned}
	p_J^{\text{out}} = 1-\frac{\Gamma(n_t,\frac{(2^R-1)n_t}{P})}{\Gamma(n_t)} + \frac{e^{\frac{1}{P_J}}\Gamma(n_t,\frac{(2^R-1)n_t}{P}+\frac{1}{P_J})}{\Gamma(n_t)(1+\frac{P}{n_t(2^R-1)P_J})^{n_t}}. \label{eq:misoach2}
	\end{aligned}
\end{equation}	

\textit{Case 2 (SIMO $(n_r \times 1)$):}
In this case, the input-output relation is given by the following expression
\begin{equation}
	\begin{aligned}
	\mathbf{y}= \mathbf{h}x+\mathbf{h_J} {x_J}+\textbf{n}
	\end{aligned}
\end{equation}
As the transmitter does not have CSI, it allocates all the power to the message signal. The instantaneous achievable rate for this case is
\begin{equation}
	\begin{aligned}
	R_i= \log_2 \bigg( 1+\frac{|\mathbf{h}|^2P}{1+|{h_J}|^2 P_J}  \bigg).
	\end{aligned}
\end{equation} 
In the above, $ |\mathbf{h}|^2 = h_1^2 + h_2^2 + \cdots h_{n_{r}}^2 $ follows chi-squared distribution with $2{n_{r}}$ degrees of freedom, i.e. $|\mathbf{h}|^2 \sim \chi^2_{2n_{r}} $. The outage probability is obtained as follows:
	\begin{equation} 
	\begin{split}
	p_{J}^{\text{out}} &= Pr\{ R_i < R \}, \\
	&= Pr\big\{\log_2 \Big( 1+\frac{|\mathbf{h}|^2P}{1+|{h_J}|^2 P_J}  \Big)<R \big\},  \\
	&= Pr\big\{|\mathbf{h}|^2P - (2^R-1)|{h_J}|^2 P_J  <2^R-1 \big\} .
	\end{split}
	\end{equation}
Following similar approach as in (\ref{eq:misopout}) (previous case), the outage probability expression becomes
	\begin{equation}
	\begin{aligned}
p_{J}^{\text{out}} = 1-\frac{\Gamma(n_{r},\frac{(2^R-1)}{P})}{\Gamma(n_{r})} + \frac{e^{\frac{1}{P_J}}\Gamma(n_{r},\frac{(2^R-1)}{P}+\frac{1}{P_J})}{\Gamma(n_{r})(1+\frac{P}{(2^R-1)P_J})^{n_{r}}}. \label{eq:simoach2}
	\end{aligned}
	\end{equation}
\textit{Case 3 (MIMO $(n_r \times n_t)$):}	
The instantaneous achievable rate for a MIMO fading channel using Alamouti coding \cite{sandhu2000space,tse2004fundamentals} in the presence of jammer is given by
\begin{equation}\label{MIMO_ala_cap}
R_i = \frac{K}{T}\log_2\Big(1+\frac{\frac{P}{n_t}||\mathbf{H}||_F^2}{1+|{h_J}|^2P_J} \Big).
\end{equation}
To simplify the notations, define the following quantities $ V \triangleq \frac{P}{n_t} ||\mathbf{H}||_F^2 $ and $S \triangleq ~|{h_J}|^2$. Note that $V$ is said to have the gamma distribution with $\Omega=N=n_tn_{r}$ and $\beta=\frac{P}{n_t}$, i.e. $V \sim G(\Omega, \beta)$, and S is exponentially distributed $\sim e^{-s}$. 
The quantity $\frac{K}{T}$ denotes the rate of the space-time block coding (STBC), where $K$ represents the number of symbols transmitted, and $T$ denotes the number of time slots used for transmission. The rate of the STBC is considered to be $1$ in this case. However, we can extend the result to any rate of the Alamouti coding scheme. The outage probability of a MIMO system with Alamouti coding at the transmitter in the presence of jammer is given as follows
\begin{equation}
\begin{split}
p_{J}^{\text{out}} &  =Pr( R_i < R),\nonumber\\
&=  Pr \Big( \log_2\Big(1+\frac{V}{1+{SP_J}}\Big)<R \Big), \nonumber  \\
\end{split}
\end{equation}
\begin{equation}
\begin{split}
&= Pr\Big(\frac{V-(2^R-1)}{P_J(2^R-1)}< S\Big),\nonumber \\  
&= 1- Pr\Big(S <\frac{V-(2^R-1)}{P_J(2^R-1)}\Big),\nonumber \\
\end{split}
\end{equation}
\begin{equation}
\begin{split}
&= 1- \int_{v=(2^R-1)}^\infty \int_{s=0}^{\frac{v-(2^R-1)}{P_J(2^R-1)}} f_V(v) f_S(s) ds dv.  \label{pout_mimo_ala_jam}
\end{split}
\end{equation}
Note that 
\begin{equation}\label{pdf_gamma}
f_V(v)=  \dfrac{v^{(N-1)}e^{\frac{-v}{\beta}}}{\Gamma(N) \beta^N}~ \text{and}~ f_S(s)=e^{-s}.
\end{equation}
Substituting \eqref{pdf_gamma} in \eqref{pout_mimo_ala_jam} and simplifying further, results in the following expression


\begin{equation}
p_{J}^{\text{out}} = 1- \frac{\Gamma(N, \frac{(2^R-1)}{\beta})}{\Gamma(N)} + \dfrac{e^{\frac{1}{P_J}} \Gamma\bigg(N,\big( \frac{2^R-1}{\beta} + \frac{1}{P_J}\big)\bigg)}{\beta^N  \Gamma(N) \big[ \frac{1}{\beta} + \frac{1}{(2^R-1)P_J}\big]^N}.\label{eq:mimoach2}
\end{equation}\vspace{-8pt}
\end{proof}
\vspace{-20pt}
\subsection{Asymptotic Analysis}
To obtain further insights on the role of antennas on the outage probability in the presence of jammer, the power at the transmitter and jammer are driven to infinity but their ratio is held constant, i.e., $\frac{P_J}{P} = \eta$. It is assumed that $\eta \geq 1$  and this condition ensures that when $P \rightarrow \infty$, then $P_J \rightarrow \infty$. The asymptotic expressions for different MIMO configurations are as follows:
\begin{align}
& \lim_{P,P_J\longrightarrow \infty} p_J^{\text{out}}  =  \frac{1}{(1+\frac{1}{\eta n_t(2^R-1)})^{n_{t}}}, \label{eq:miso_asym_wj} \qquad \text{(MISO)}\\
& \lim_{P,P_J\longrightarrow \infty}p_{J}^{\text{out}}  =    
 \frac{1}{(1+\frac{1}{\eta(2^R-1)})^{n_{r}}}. \label{eq:simo_asym_wj} \qquad \text{(SIMO)} \\
\text{and }& \lim_{P,P_J\longrightarrow \infty}p_{J}^{\text{out}}  =  \dfrac{1}{\big[ 1 + \frac{1}{\eta n_t(2^R-1)}\big]^N}. \label{eq:mimo_asym_wj} \qquad \text{(MIMO with Alamouti)}
\end{align}
\textbf{Remarks:}
\begin{enumerate}
	\item As a special case, the outage probability without jamming for different MIMO configurations can be obtained by setting $P_J =0$ using the results in Theorem~\ref{th:firsttheorem}. 
	\item When $\eta$ is unbounded in (\ref{eq:miso_asym_wj}), (\ref{eq:simo_asym_wj}), and (\ref{eq:mimo_asym_wj}), the outage probability becomes $1$ for all the cases, irrespective of the number of antennas at the transmitter, receiver, or both. 
	\item As a special case we can obtain the outage probability for the SISO with jammer by substituting $n_t = 1$ or $n_r=1$ in the result for MISO or SIMO in \eqref{eq:th1eq1} or \eqref{eq:th1eq2}, respectively.  By setting $n_r= 1$ in Theorem~\ref{th:firsttheorem}, one can obtain the outage probability for the Alamouti space-time code in case of the MISO under a jamming attack.
	\item The developed results in Theorem~\ref{th:firsttheorem} and (\ref{eq:miso_asym_wj})-(\ref{eq:mimo_asym_wj}) can be useful to study the benefits of multiple antenna techniques in achieving desired performance with less power budget under jamming attack. In particular, it is found that the MIMO system with Alamouti coding can provide superior performance at less power budget in comparison to other antenna configurations such as MISO (See Figs~\ref{subfig-1:figure3} and \ref{subfig-1:figure5} in Sec~\ref{sec:result}).  
\end{enumerate} 
\section{Stability Analysis for Multi-antenna System under Jamming Attack}	
In this section, the average service rate of the point-to-point MIMO system with random arrival of data at the transmitter under jamming attack is derived. As the attacker has energy harvesting ability, the impact of jamming on the service rate is affected by the energy arrival rate at the jammer and the capacity of the battery. The results are derived for the battery with unlimited and limited capacity at the jammer. The outage probabilities derived in the previous section are used for determining the successful decoding of the data at the receiver. 
\vspace{-8pt}
\subsection{Battery with unlimited capacity}
To determine the average service rate in this case, it is required to model how the state of the energy buffer at the jammer evolves. The evolution of the energy buffer can be described by the Markov chain, as shown in Fig.~\ref{fig:markov-chain}. The evolution of the Markov chain is characterized by the probability of jamming $(p_J)$ and the energy arrival rate at the jammer $(\delta)$. To determine the average service rate with infinite energy buffer at the jammer, the following cases are considered: (a) probability of jamming is more than the energy arrival rate at the jammer, and (b) probability of jamming is less than the energy arrival rate at the jammer. The reason for considering these two cases is explained below. In the first case, the system is limited by its harvested energy; thus, it is not always available to create interference, even if needed. In the second case, since the jamming probability is less than the energy arrival rate, the jammer is operating as if it was connected to the power grid without energy limitations.

The average service rate with infinite energy buffer at the jammer is stated in the following theorem. 
\vspace{-7pt}
\begin{figure}[h]
	\centering
	\includegraphics[trim={0cm 0cm 0cm 0cm},clip, height=.7 in,width=3.4 in]{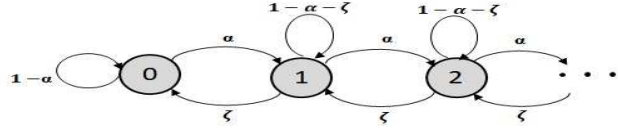}
	\setlength{\abovecaptionskip}{-1pt}
	\setlength{\belowcaptionskip}{-12pt} %
	\caption{State transition diagram of infinite capacity  energy buffer at jammer, where $\zeta=(1-\delta)p_J$, and $\alpha=\delta(1-p_J)$.} \label{fig:markov-chain}
\end{figure}
\begin{theorem}\label{th:secondtheorem}
The average service rate of the system when the jammer has infinite energy buffer is as follows:
\begin{enumerate}
	\item When $p_J \geq \delta$
\begin{equation}
\mu  = \hspace{1mm} \bigg(1- \frac{(1-p_J)\delta}{(1-\delta )}\bigg)(1 - p_{WoJ}^{\text{out}}) +  \bigg(\frac{(1-p_J)\delta}{(1-\delta )}\bigg) (1- p_J^{\text{out}}).\label{eq:th2eq1}
\end{equation}
\item When $p_J < \delta$
\begin{equation}
	\mu  = (1-p_J)(1 - p_{WoJ}^{\text{out}}) + \hspace{1mm} p_J(1- p_J^{\text{out}}). \label{eq:th2eq2}
	\end{equation}
\end{enumerate}
The average service rate for MISO, SIMO, and MIMO with Alamouti coding can be obtained by using the results on outage probability from Theorem~\ref{th:firsttheorem}. When $\lambda < \mu$, the system is said to be stable. 
\end{theorem}
\begin{proof}
To determine the average service rate in the presence of an energy harvesting jammer, the following cases are considered. \\
\textit{Case 1 (When jammer is sending energy packets at a rate ($ p_J$) greater than its arrival rate ($\delta $)): }
If $\delta \leq p_J$, then the jammer may not always have sufficient energy to disrupt the communication, and the jammer's energy buffer can be empty. To characterize the service rate, it is required to determine the probability that the energy buffer is empty, i.e., $Pr(B=0)$, where $B$ denotes the size of the energy buffer.
 From Fig.~\ref{fig:markov-chain}, one can see that the steady-state distribution need to satisfy the following
\begin{equation}
\begin{split}
& \zeta \pi_{i+1} = \alpha \pi_{i} ~~~~\text{for} ~~ 0 \leq i \leq B-1. \nonumber \\
\end{split}
\end{equation}
Normalizing the probabilities, following is obtained
 \begin{equation} 
 \begin{split}
\pi_{0} \sum_{i=0}^{\infty} \big(\frac{\alpha}{\zeta}\big)^i =1,~ 
\text{or } \pi_{0} = 1- \frac{\alpha}{\zeta},
\end{split}
\end{equation}
where $\pi_{i}$ is the probability of being in the $i^{th}$ state. Therefore,  probability that the jammer energy buffer can be empty is given by
\begin{equation}
\pi_{0}=Pr(B=0) = 1- \frac{(1-p_J)\delta}{(1-\delta )p_J}=1-\frac{\alpha}{\zeta}.
\end{equation}
Hence, the probability that energy buffer is not empty is given by the following
\begin{equation}
Pr(B\neq0) =  \frac{(1-p_J)\delta}{(1-\delta )p_J} =\frac{\alpha}{\zeta}.
\end{equation}
Using the above, the service rate of the system becomes
\begin{equation}
\begin{split}
&\mu= (1-p_J Pr(B\neq0)) (1 - p_{WoJ}^{\text{out}}) ~ + ~ p_J Pr(B\neq0) (1- p_J^{\text{out}}), \\
&= \bigg(1- \frac{(1-p_J)\delta}{(1-\delta )}\bigg)(1 - p_{WoJ}^{\text{out}}) ~  + ~ \bigg(\frac{(1-p_J)\delta}{(1-\delta )}\bigg) (1- p_J^{\text{out}}).\\
&= \bigg(1- \frac{\alpha}{(1-\delta )}\bigg)(1 - p_{WoJ}^{\text{out}}) ~  + ~ \bigg(\frac{\alpha}{(1-\delta )}\bigg) (1- p_J^{\text{out}}).
\end{split}\label{eq:serratecase-2}
\end{equation}
Note that the outage probability without jamming $(p_{WoJ}^{\text{out}})$ can be obtained by setting $P_J=0$ for different antenna configurations in Theorem~\ref{th:firsttheorem}.\\
\textit{Case 2 (When jammer is sending energy packets at a rate ($ p_J$) less than its arrival rate ($\delta $)):} If $\delta \geq p_J$, then the jammer always have energy to send the jamming signal and the probability that the energy buffer is empty is $0$, i.e., $Pr(B=0)=0$. Hence, the service rate of the system can be written as:
\vspace{-8pt}
\begin{align}
& \mu = \hspace{1mm}(1-p_J)\hspace{1mm}(1 - p_{WoJ}^{\text{out}}) + \hspace{1mm} p_J(1- p_J^{\text{out}}). \label{eq:serratecase-1}
\end{align}\vspace{-20pt}
\end{proof} \vspace{-12pt}
\subsection{Battery with Finite Capacity}
Another important factor that determines the efficiency of the jammer is the capacity of the energy buffer. This section characterizes the service rate in the case of the finite energy buffer. Similar cases are considered as that in the case of a battery with unlimited capacity. Note that even when the energy arrival rate can be high at the jammer but limited battery capacity can hinder the jamming ability of the attacker. The evolution of the energy buffer of size $B$ at the jammer can be described by the Markov chain, as shown in Fig.~\ref{fig:markov-chain_finite}. The evolution of the Markov chain is characterized by the probability of jamming and the energy arrival rate at the attacker.
\begin{figure}[h]
	\centering
	\includegraphics[trim={0cm 0cm 0cm 0cm},clip, height=.9 in,width=3.4 in]{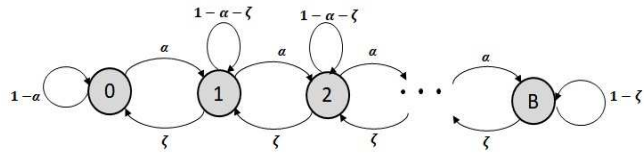}
	\setlength{\abovecaptionskip}{-6pt}
	\setlength{\belowcaptionskip}{-15pt} %
\caption{State transition diagram of finite energy buffer at the jammer, where $\zeta=(1-\delta)p_J$, and $\alpha=\delta(1-p_J)$.}\label{fig:markov-chain_finite}
\end{figure}

The average service rate with finite energy battery at the jammer is stated in the following theorem.

\begin{theorem}\label{th:thirdtheorem}
The average service rate of the system when the jammer has finite energy buffer is as follows:
	\begin{enumerate}
		\item When $p_J \geq \delta$
		\begin{align}
		&\mu  =  \bigg(1- p_J\bigg(1-\frac{(p_J-\delta){\zeta}^B}{\zeta^{B+1}-\alpha^{B+1}}\bigg)\bigg)(1 - p_{WoJ}^{\text{out}}) +~p_J\bigg(1-\frac{(p_J-\delta){\zeta}^B}{\zeta^{B+1}-\alpha^{B+1}}\bigg) (1- p_J^{\text{out}}).\label{eq:th3eq1}
		\end{align}
		\item When $p_J < \delta$
		\begin{align}
		& \mu = (1-p_J)(1 - p_{WoJ}^{\text{out}}) +  p_J(1- p_J^{\text{out}}). \label{eq:th3eq2}
		\end{align}
	\end{enumerate}
	Using the results on outage probability from Theorem~\ref{th:firsttheorem}, the average service rate for MISO, SIMO, and MIMO with Alamouti coding can be obtained. For $\lambda <\mu$, the system is stable.	
\end{theorem}
\begin{proof}
	To determine the average service rate in the presence of an energy harvesting jammer with finite battery, the following cases are considered. \\
	\textit{Case 1 (When jammer is sending energy packets at a rate ($ p_J$) greater than its arrival rate ($\delta $)): }
	If $\delta \leq p_J$, then the jammer may not have sufficient energy to disrupt the ongoing communication as its energy buffer can be empty. In this case,  the probability that energy buffer at the jammer is empty depends also on the capacity of the energy buffer. This probability is obtained by solving the balance equation of the Markov chain for finite battery and is given by
	\begin{equation}
	Pr(B=0) = \frac{(p_J-\delta){\zeta}^B}{\zeta^{B+1}-\alpha^{B+1}}.
	\end{equation}
	Hence, the probability that jammer energy buffer is not empty is
	\begin{equation}
	Pr(B\neq0) =  1-\frac{(p_J-\delta){\zeta}^B}{\zeta^{B+1}-\alpha^{B+1}}.
	\end{equation}
	Then, the service rate of the system is
	\begin{align}
	& \mu  =  \bigg(1- p_J\bigg(1-\frac{(p_J-\delta){\zeta}^B}{\zeta^{B+1}-\alpha^{B+1}}\bigg)\bigg)(1 - p_{WoJ}^{\text{out}}) +~p_J\bigg(1-\frac{(p_J-\delta){\zeta}^B}{\zeta^{B+1}-\alpha^{B+1}}\bigg) (1- p_J^{\text{out}}).
	\end{align}
\textit{Case 2 (When jammer is sending energy packets at a rate ($ p_J$) less than its arrival rate ($\delta $)):} If $\delta > p_J$, then the jammer always have energy to send the jamming signal, i.e., $Pr(B=0)=0$. 
Hence, the service rate of the system can be written as given in ( \ref{eq:th3eq2}). 
\end{proof}
\textit{\textbf{Remarks:}}
\begin{enumerate}
\item When $\delta > p_J$, it is observed from \eqref{eq:th2eq2} and \eqref{eq:th3eq2} that the performance of battery with finite capacity and infinite capacity are same. In this case, having a battery with infinite capacity does not help the jammer. 

\item As a special case, one can obtain the service rate for different MIMO configurations without jamming by substituting $p_J = 0$ either in a finite or infinite battery capacity case. To the best of the authors' knowledge, the service rate of a MIMO system with an Alamouti coding scheme has not been characterized in existing results even without jamming. 

\item When $\delta = 1$ and $p_J = 1$, one can obtain the service rate of the considered system model under jamming attack, where the jammer always disrupts the communication.
\end{enumerate}
\vspace{-6pt}
\section{Delay and AAoI Analysis under Jamming Attack}
For the considered system model with bursty traffic, besides service rate, the delay is another important metric when the transmitter has time-sensitive data. Furthermore, AoI, which is a more general form of latency, is also important since, in status updating systems (common in IoT or cyber-physical systems), the freshness of information is crucial. As jamming is one of the common DoS (Denial of Service) attacks, it is important to evaluate these metrics in such cases. It is desirable to have a low delay and small AoI. The AoI captures the freshness in information, whereas the delay captures the queuing delay and transmission delay between the transmitter and the receiver. In the following, average delay and average AoI (AAoI) are characterized for different MIMO configurations. The developed results help to investigate the interplay between delay and AoI for multi-antenna configurations when the attacker has a battery with finite or infinite capacity. To the best of the authors' knowledge, the role of multiple antennas in minimizing delay and AAoI under jamming attack has not been explored in the existing literature.
\vspace{-10pt}
\subsection{Average Delay} 
The delay consists of two components, the queuing delay and the transmission delay from the transmitter to the receiver. The average transmission delay $(D_T)$ is inversely proportional to the average service probability and is given by
\begin{equation}
	D_T= \frac{1}{\mu}. \label{eq:delay1]}
\end{equation}
To determine the queuing delay, it is required to use the average queue length which is given by
\begin{equation} 
\begin{split}
	\bar{Q} & = \frac{\lambda (1-\lambda)}{\mu - \lambda}. \label{eq:delay2}
	\end{split}
\end{equation}
Using Little's theorem, the queuing delay $(D_Q)$ is expressed as follows
\begin{equation}
	D_Q= \frac{\bar{Q}}{\lambda}.  \label{eq:delay3}
\end{equation}
The average packet delay (D) using (\ref{eq:delay2}) and (\ref{eq:delay3}) becomes
	\begin{equation}
	\begin{split}
	D & = D_T + D_Q = \frac{1}{\mu } +\frac{1-\lambda}{\mu - \lambda}. \label{eq:average_delay} 
	\end{split}
	\end{equation}
One can obtain the average packet delay when the jammer has infinite and finite battery capacity using the service rate expressions from Theorems~\ref{th:secondtheorem} and \ref{th:thirdtheorem} in \eqref{eq:average_delay}, respectively. 
\vspace{-10pt}
\subsection{Average Age of Information (AAoI)} 
The AAoI metric captures the freshness of information about a source at a remote destination. The objective here is to understand the effect of the multiple antennas at the transmitter and receiver on the AAoI under jamming attack. The considered queuing model in this paper is a Geo/Geo/1 queue, and the AAoI is given by 
\begin{equation}\label{AAoI} 
\begin{split}
AAoI & =\frac{1}{\lambda}+\frac{1 - \lambda}{\mu-\lambda}- \frac{\lambda}{{\mu}^2}+\frac{\lambda}{\mu}.
\end{split}
\end{equation}
The proof for the previous expression can be found in \cite{8764468}. One can obtain the AAoI when the attacker has infinite and finite battery capacity using the service rate expressions from  Theorems~\ref{th:secondtheorem} and \ref{th:thirdtheorem} in \eqref{AAoI}, respectively. We would like to emphasize that the purpose here is to study the effect of multiple antennas, jamming attacks, and the Alamouti coding scheme on the AAoI and not to derive the metric from scratch. Furthermore, the results obtained in this work can be utilized in other models for AAoI, for example, at the generate at-will policy \cite{8815559} or in different queuing disciplines with or without packet management \cite{8764468}.
\vspace{-10pt}
\subsection{Minimization of AAoI}
In many applications, where the transmitter has time-critical data, it is required to minimize average delay and AAoI. It is desirable to support a high arrival rate even under a jamming attack. However, Fig~\ref{subfig-1:figure10} shows that a low or high arrival rate can increase the AAoI at the receiver. Furthermore, this is in contrast to the behavior of the delay against arrival rate in Fig~\ref{subfig-1:figure9}. Minimizing delay does not necessarily imply minimization of AAoI and vice-versa. This motivates to minimize the AAoI with respect to arrival rate for different antenna configurations under jamming attack. When there is also a delay constraint on the traffic, the arrival rate that minimizes AAoI can be different from the solution to the optimization problem, where it is only required to minimize the AAoI without any delay constraint. To get further insights into this problem, the following optimization problems are considered.
\subsubsection{Minimization of AAoI for Delay-Tolerant System} In this case, the objective is to minimize AAoI with respect to the arrival rate provided the queue remains stable and there is no constraint on the delay. The optimization problem is stated in the following.
\begin{equation}\label{AAoI_min_state}
\begin{aligned}
& \underset{\lambda}{\text{minimize}}
& & \mathrm{AAoI} ~ \text{in} ~\eqref{AAoI},  \quad \text{such that} \quad  \lambda < \mu.
\end{aligned}
\end{equation}
The constraint ensures that the queue remains stable. We can solve this problem as follows.
\begin{equation}
\begin{aligned}
& \frac{\partial AAoI}{\partial \lambda} = 0,
\end{aligned}
\end{equation}
\begin{equation}\label{AAoI_opt_1}
\begin{aligned}
\text{or } & \frac{\mu-1}{\mu^2}+\frac{1-\mu}{(\mu-\lambda)^2} -\frac{1}{\lambda^2}= 0,\\ 
\text{or } &\lambda^4(\mu-1)+\lambda^3(2\mu-2\mu^2)-\lambda^2\mu^2+2\lambda\mu^3-\mu^4=0.\\ 
\end{aligned}
\end{equation}
The value of $\lambda$ is chosen which satisfies the above equation \eqref{AAoI_opt_1} and also, does not violate the stability condition of the queue. 

\subsubsection{Minimization of AAoI for Delay-Sensitive System} In this case, it is required to minimize the AAoI with respect to the arrival rate provided the average delay in the system should not exceed a threshold ($D_{th}$) and the queue remains stable. The optimization problem is stated in the following:
\begin{equation}\label{AAoI_opt_2}
\begin{aligned}
& \underset{\lambda}{\text{minimize}}
& & \mathrm{AAoI} ~ \text{in} ~\eqref{AAoI},
~\text{such that} ~ \lambda < \mu, \text{ and } D \leq D_{th} ~\text{in} ~\eqref{eq:average_delay}.
\end{aligned}
\end{equation}
The above optimization problem is solved numerically and discussed in Sec~\ref{sec:result}. For both the optimization problems, the service rate depends on the number of antennas at the legitimate nodes,  battery capacity at the jammer, and other parameters of the system (See Theorem~\ref{th:secondtheorem}).
\vspace{-4pt}
\section{Stability Region for the 2-user SIMO Broadcast Channel}
 \begin{figure}[!ht]
	\centering
	\includegraphics[trim={0cm 0cm 0cm 0cm},clip, height=1.3in,width=3.4 in] {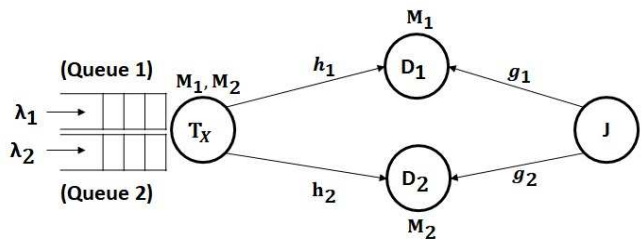}
	\setlength{\belowcaptionskip}{-12pt}
	\caption{2-user broadcast channel with the jammer ($J$).}\label{fig:Broadcast}
\end{figure}
 In this section, it is shown how the results developed for the point-to-point system  help to characterize the stability region for the 2-user SIMO broadcast channel in the presence of energy harvesting jammer, where the receiver has $n_r$ antennas. The transmitter has two queues to store the incoming traffic and the packet stored at $i^{\text{th}}$ queue needs to be delivered to the $D_i^{\text{th}}$ receiver (See Fig.~\ref{fig:Broadcast}). Each receiver decodes its intended packet by treating other user's messages as noise. Due to lack of space, the stability region is characterized for the case when jammer has unlimited capacity and $\delta>p_j$. In this case, jammer energy queue is non-empty. The received signal at $i^{th}$ ($i \in \{1,2\}$) receiver is given by
 \begin{equation}
 	\mathbf{y}_i = \mathbf{h}_i x +g_i x_J + z_i,  
 \end{equation}
 where $\mathbf{h}_i \in \mathcal{C}^{n_r \times 1}$  is the channel between transmitter and the $i^{\text{th}}$ receiver and it undergoes Rayleigh fading. Hence, $|\mathbf{h}_i|^2$ follows Chi-Square distribution with $2n_r$ degrees of freedom. The channel $g_i \in \mathcal{C}$ is the channel between the attacker and $i^{\text{th}}$ receiver. The channel $g_i$ also undergoes Rayleigh fading and hence, $|g_i|^2$ follows an exponential distribution. The noise at receiver~$i$ is distributed as $z_i\sim \mathcal{CN}(0, 1)$. The input signal sent to the channel is $x = x_1 + x_2$ (when both queues are non-empty) or $x = x_i$ (when the only ith queue has a packet to send). 
 
 The packet arrival processes at the first and the second queue are assumed to be independent and stationary with  mean rates $\lambda_1$ and $\lambda_2$ in packets per slot, respectively. Both queues have an infinite capacity to store incoming packets and $\bar{Q}_i$ denotes the size in the  number of packets of the i-th queue. All other assumptions remain the same as that of the point to point MIMO model. The average service rate at the queues are given by
 \vspace{-2pt}
 \begin{align}
 & \mu_1=Pr(\bar{Q}_2>0) Pr(D_{1/1,2}) + Pr(\bar{Q}_2 = 0) Pr(D_{1/1}),  \label{broadcast_mu_1} \\
 & \mu_2=Pr(\bar{Q}_1>0) Pr(D_{2/1,2})+Pr(\bar{Q}_1=0) Pr(D_{2/2}). \label{broadcast_mu_2}
 \end{align}
 where $D_{i/\tau}$ denote the event that receiver~$i$ is able to decode the packet transmitted from the i-th queue of the transmitter given a set of non-empty queues denoted by $\tau$ ($\tau\in \{1, 2\}$).  Due to the interaction between the queues, the stochastic dominance technique \cite{21216} is used. There are two dominant systems. In the first dominant system, when the first queue is empty, then the source transmits a dummy packet for receiver~$1$, while the second queue behaves in the same way as in the original system. In the second dominant system, when the second queue is empty, then the source transmits a dummy packet for receiver~$2$, while the first queue behaves in the same way as in the original system.  Following the process, as given in \cite{8320826}, the stability region for both the dominant systems are 
  \begin{figure*}	
  \begin{align}
   &\mathcal{R}_1 = \left\{(\lambda_1,\lambda_2): \frac{\lambda_1}{Pr(D_{1/1})}+ \frac{Pr(D_{1/1})-Pr(D_{1/1,2})}{Pr(D_{1/1})Pr(D_{2/1,2})}\lambda_2<1, \lambda_2<Pr(D_{2/1,2})\right\}, \label{rate_1} \\
 &\mathcal{R}_2 = \left\{(\lambda_1,\lambda_2): \frac{\lambda_2}{Pr(D_{2/2})}+ \frac{Pr(D_{2/2})-Pr(D_{2/1,2})}{Pr(D_{2/2})Pr(D_{1/1,2})}\lambda_1<1, \lambda_1<Pr(D_{1/1,2})\right\}. \label{rate_2}
\end{align}
\end{figure*} 
 given by $\mathcal{R} = \mathcal{R}_1 \cup \mathcal{R}_2$, where $\mathcal{R}_1$ and $\mathcal{R}_2$ are given by \eqref{rate_1} and \eqref{rate_2}, respectively. To determine the stability region, it is required to determine the various success probabilities. These are obtained using the outage probabilities derived for SIMO in Section~\ref{sec:outage}.  The ratio of signal power to noise power along with jamming and/or interference at receiver~$i$ is given as
 \begin{equation}
 SJINR_i = \frac{P_i |\mathbf{h}_i|^2}{1 + P_J |g_i|^2 + P_j |\mathbf{h}_i|^2},\text{and } SJNR_i =  \frac{P_i |\mathbf{h}_i|^2}{1 + P_J |g_i|^2}. \\
\end{equation}
where $i, j \in \{1, 2\}$  and $ i \neq j$.\\
Similarly, $SINR_i$ (signal power to interference and noise power ratio) and $SNR_i$ (signal to noise ratio) at receiver $i$ are denoted as follows
\begin{equation}
SINR_i = \frac{P_i |\mathbf{h}_i|^2}{1 + P_j |\mathbf{h}_i|^2}, \text{ and } SNR_i = P_i |\mathbf{h}_i|^2. \\ 
\end{equation} 
where $i, j \in \{1, 2\}$ and $ i \neq j$. The success probabilities for both the receivers are obtained with and without jamming for different status of the queue in the following. 
\subsubsection{When $\bar{Q}_1=0$ and $\bar{Q}_2\neq0$} In this case, the second queue at the transmitter sends a packet for receiver 2 as $\bar{Q}_1=0$. The receiver 2 can decode its intended packet if one of the following events occur: 
 $ D_{2/2}^{\text{WoJ}} = \left\{ SNR_2\geq\gamma_2 \right\}$  (without jamming) and  $D_{2/2}^{\text{WJ}}=\left\{SJNR_2\geq\gamma_2\right\}$ (with jamming). The quantity $\gamma_i$ corresponds to the decoding threshold for successful decoding of the packet at receiver~$i$ and is related to the rate by the following relation: $\gamma_i = 2^{R_{i}}-1$ ($i$=1,2). It is not difficult to see that the $SJNR_i$ is similar to the SJNR in case of point-to-point model. Following a similar procedure as in Theorem~\ref{th:firsttheorem} (using \eqref{use_SuccPr}), the success probabilities can be obtained as given below. 
\begin{equation}
Pr\big(D_{2/2}^{\text{WoJ}}\Big)=Pr(SNR_2\geq\gamma_2)
=Pr\Big(|\mathbf{h}_2|^2\geq\frac{\gamma_2}{P_2}\Big)
=\frac{\Gamma(n_{r},\frac{\gamma_2}{P_{2}})}{\Gamma(n_{r})}.
\label{WoJ}
\end{equation}
Similarly, one can show the following
\begin{align}
&Pr\big(D_{2/2}^{\text{WJ}}\big)=Pr(SJNR_2\geq\gamma_2)
= Pr\Big(|g_2|^2\leq\frac{P_2 |\mathbf{h}_2|^2-\gamma_2}{\gamma_2 P_J }\Big),\nonumber\\
&=\frac{\Gamma(n_{r},\frac{\gamma_2}{P_{2}})}{\Gamma(n_{r})}-\frac{e^{\frac{1}{P_J}}\Gamma(n_{r},\frac{\gamma_2}{P_{2}}+\frac{1}{P_J})}{\Gamma(n_{r})(1+\frac{P_{2}}{\gamma_2 P_J})^{n_{r}}}.\label{WJ}
\end{align}
Using \eqref{WoJ} and \eqref{WJ}, the success probability at receiver~$2$ is given by
 \begin{align}
 &Pr(D_{2/2}) = (1-p_J)Pr\big( D_{2/2}^{\text{WoJ}}\big)+p_J Pr\big(D_{2/2}^{\text{WJ}}\big),\\
& = (1-p_J) \frac{\Gamma(n_{r},\frac{\gamma_2}{P_{2}})}{\Gamma(n_{r})}+p_J\Bigg(\frac{\Gamma(n_{r},\frac{\gamma_2}{P_{2}})}{\Gamma(n_{r})}-\frac{e^{\frac{1}{P_J}}\Gamma(n_{r},\frac{\gamma_2}{P_{2}}+\frac{1}{P_J})}{\Gamma(n_{r})(1+\frac{P_{2}}{\gamma_2 P_J})^{n_{r}}}\Bigg).
\end{align}
 \subsubsection{When $\bar{Q}_1\neq0$ and $\bar{Q}_2=0$} In this case, the first queue at the transmitter sends a packet for the receiver 1 as $\bar{Q}_2=0$. The receiver 1 can decode its intended packet if the following events are true.
 \begin{align} 
 &D_{1/1}^{\text{WoJ}}=\left\{ SNR_1\geq\gamma_1 \right\} ~\text{without jamming}, \nonumber\\
 &D_{1/1}^{\text{WJ}}=\left\{SJNR_1\geq\gamma_1\right\}~\text{with jamming.}
 \end{align}
Following a similar procedure as in Theorem~\ref{th:firsttheorem} (using \eqref{use_SuccPr}), the success probabilities are given as follows
 \begin{align}
& Pr(D_{1/1}) = (1-p_J)Pr\big( D_{1/1}^{\text{WoJ}}\big)+p_J Pr\big(D_{1/1}^{\text{WJ}}\big), \\
 &=(1-p_J)\frac{\Gamma(n_{r},\frac{\gamma_1}{P_{1}})}{\Gamma(n_{r})}+p_J\Bigg(\frac{\Gamma(n_{r},\frac{\gamma_1}{P_{1}})}{\Gamma(n_{r})}-\frac{e^{\frac{1}{P_J}}\Gamma(n_{r},\frac{\gamma_1}{P_{1}}+\frac{1}{P_J})}{\Gamma(n_{r})(1+\frac{P_{1}}{\gamma_1 P_J})^{n_{r}}}\Bigg). 
\end{align}
\subsubsection{When $\bar{Q}_1\neq0$ and $\bar{Q}_2\neq0$} In this case, both queues send packets. The packet sent by the first (second) queue should be decoded at receiver 1 (receiver 2), which are represented by the following events
\begin{align} 
&D_{1/1,2}^{\text{WoJ}}=\left\{ SINR_1\geq\gamma_1 \right\} ~\text{without jamming},\nonumber \\
&D_{1/1,2}^{\text{WJ}}=\left\{SJINR_1\geq\gamma_1\right\}~\text{with jamming.}
\end{align}
 In this case, although $SJINR_i$ corresponds to the ratio of signal power and noise power along with signal power of other user, with some rearrangement of terms in the outage probability calculation,  $SJINR_i$ can be presented in a similar form as that of the $SJNR_i$. Hence, following a similar procedure as in Theorem~\ref{th:firsttheorem} (using \eqref{use_SuccPr}), the success probabilities can be obtained as
\begin{align}
Pr(D_{1/1,2}) & = (1-p_J)Pr\big( D_{1/1,2}^{\text{WoJ}}\big)+p_J Pr\big(D_{1/1,2}^{\text{WJ}}\big),\\
&=(1-p_J)\frac{\Gamma(n_{r},\frac{\gamma_1}{P_{1}-\gamma_1 P_2})}{\Gamma(n_{r})}+~p_J\Big(\frac{\Gamma(n_{r},\frac{\gamma_1}{P_{1}-\gamma_1 P_2})}{\Gamma(n_{r})}-\frac{e^{\frac{1}{P_J}}\Gamma(n_{r},\frac{\gamma_1}{P_{1}-\gamma_1 P_2}+\frac{1}{P_J})}{\Gamma(n_{r})(1+\frac{P_{1}-\gamma_1 P_2}{\gamma_1 P_J})^{n_{r}}}\Big).
\end{align}
By following similar steps as in the case of receiver~$1$, the success probability for receiver~$2$ is given by the following expression
\begin{align}
Pr(D_{2/1,2}) &=(1-p_J)\frac{\Gamma(n_{r},\frac{\gamma_2}{P_{2}-\gamma_2 P_1})}{\Gamma(n_{r})} +~p_J\Big(\frac{\Gamma(n_{r},\frac{\gamma_2}{P_{2}-\gamma_2 P_2})}{\Gamma(n_{r})}-\frac{e^{\frac{1}{P_J}}\Gamma(n_{r},\frac{\gamma_2}{P_{2}-\gamma_2 P_1}+\frac{1}{P_J})}{\Gamma(n_{r})(1+\frac{P_{2}-\gamma_2 P_1}{\gamma_1 P_J})^{n_{r}}}\Big).
\end{align}
The events $D_{2/1,2}^{\text{WoJ}}$ and $D_{2/1,2}^{\text{WJ}}$ occurs with non zero probability if $P_2-\gamma_2 P_1>0$. The stability region can be obtained by substituting the success probabilities obtained for various cases in \eqref{rate_1} and \eqref{rate_2} and then, taking union of these two regions ($\mathcal{R}_1$ and $\mathcal{R}_2$).
\vspace{-3pt} 
\section{Results and Discussions}\label{sec:result}
In this section, numerical results are presented to illustrate the effect of jamming on the system performance for various parameter settings. 
\begin{figure*}[htt]
	\centering
	\subfloat[\label{subfig-1:figure3}]{\includegraphics[trim={0.2cm 0cm 0.2cm .5cm},clip, height=2.4 in,width=3.3 in]{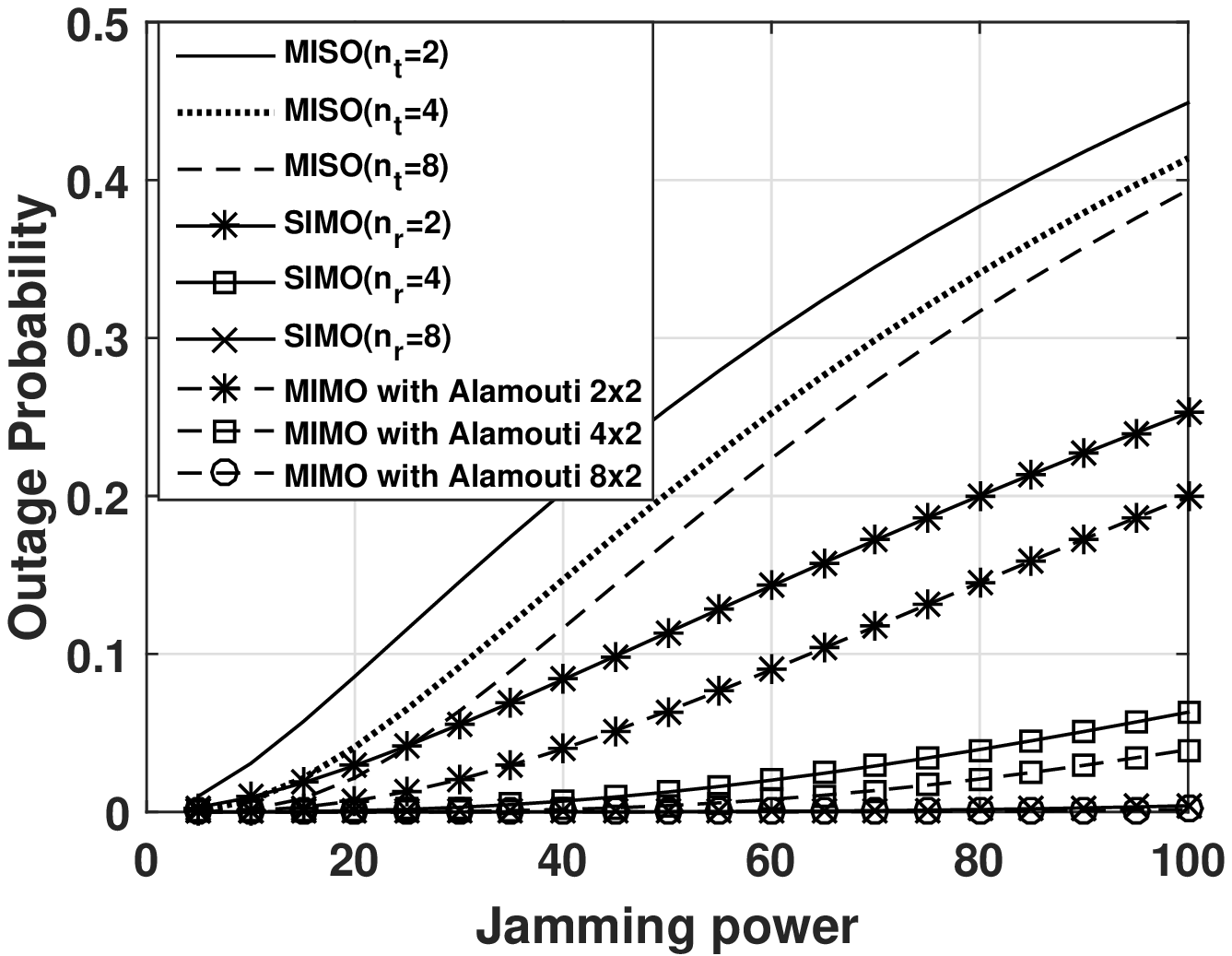}}
	\subfloat[\label{subfig-1:figure5}]{\includegraphics[trim={0.2cm 0cm 0.2cm .5cm},clip, height=2.4 in,width=3.3 in]{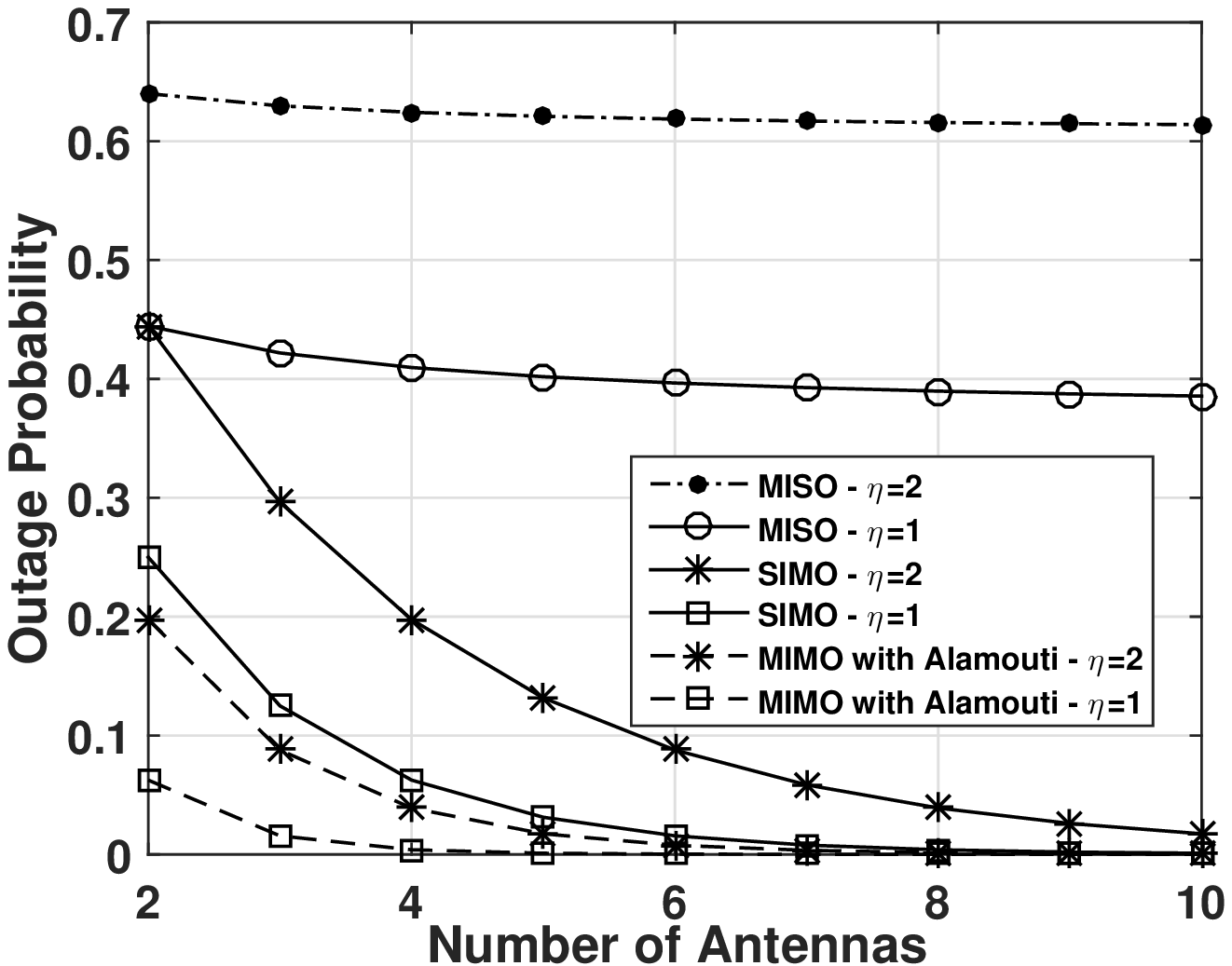}}
	\setlength{\abovecaptionskip}{-2pt}
	\setlength{\belowcaptionskip}{-15pt} %
	\caption{(a) Outage probability against jamming power with $P=20$ dB, and $R=1$; (b) Asymptotic outage probability against number of antennas $(R=1)$.}	
\end{figure*}


In Fig.~\ref{subfig-1:figure3}, the outage probability obtained in Theorem~\ref{th:firsttheorem} is plotted against jamming power for different MIMO configurations with $R=1$. Note that in this case, the attacker can always jam the  communication.\footnote{This is considered to explore how the multiple antennas help to mitigate jamming attack under Rayleigh fading scenario.} It can be observed that increasing the number of antennas at the transmitter in case of MISO does not improve the outage performance when the jamming power is increased at the transmitter. Due to the lack of CSI at the sender, the transmit diversity does not help in improving the performance. However, in the case of SIMO, the outage probability decreases with the increase in the number of antennas at the receiver, even when the jamming power increases. This improvement primarily comes from the gain in the receive diversity with the increase in the number of receive antennas. To investigate the time diversity along with space diversity, MIMO with Alamouti code is considered. It can be seen that MIMO with Alamouti code gives significant improvement compared to MISO or SIMO system.

 Fig.~\ref{subfig-1:figure5} shows the asymptotic outage probabilities obtained in \eqref{eq:miso_asym_wj}-\eqref{eq:mimo_asym_wj} against the number of antennas for a given value of $\eta$. For the MIMO with Alamouti code, the number of antennas at the transmitter is fixed at $n_t =2$, and the number of receive antennas $n_r$ is varied. From the plot, it can be observed MIMO with Alamouti code can achieve a very low value of outage probability even with a small number of antennas at the receiver or even when $\eta =2$. Due to the lack of CSI at the transmitter, outage probability is high in the case of MISO compared to SIMO or MIMO with Alamouti coding. From Fig.~\ref{subfig-1:figure3}, it can also be observed that in case of MISO system, the power budget or the number of antennas at the transmitter has to be increased significantly to meet the performance provided by SIMO or MIMO with Alamouti coding. Hence, exploiting time and space diversity can help to achieve the desired performance at a lower power budget even under jamming attack. 
\begin{figure}[h]
	\centering
	\includegraphics[trim={0.2cm 0cm 0.2cm .5cm},clip, height=2.3 in,width=3.3 in]{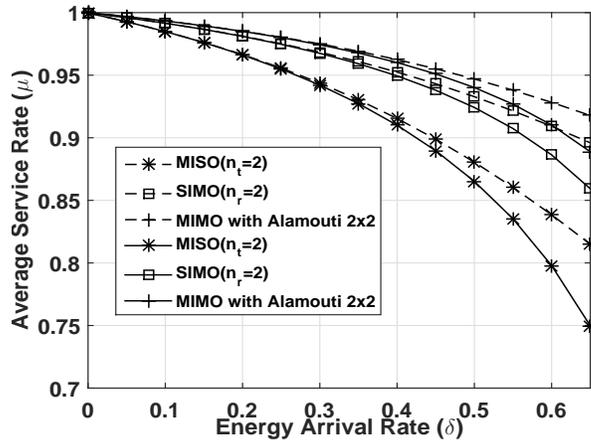}
	\setlength{\abovecaptionskip}{-4pt}
	\setlength{\belowcaptionskip}{-10pt}%
	\caption{Average service rate against energy arrival rate at jammer:
		$P=20$ dB, $P_J=20$ dB, $R=1$, $B=2$, $p_J=0.7$ and $\lambda=0.2$.}\label{fig:figure6}
\end{figure}
\begin{figure*}[!htt]
	\centering
	\subfloat[\label{subfig-1:figure7}]{\includegraphics[trim={0.2cm 0cm 0.2cm .5cm},clip, height=2.4 in,width=3.3 in]{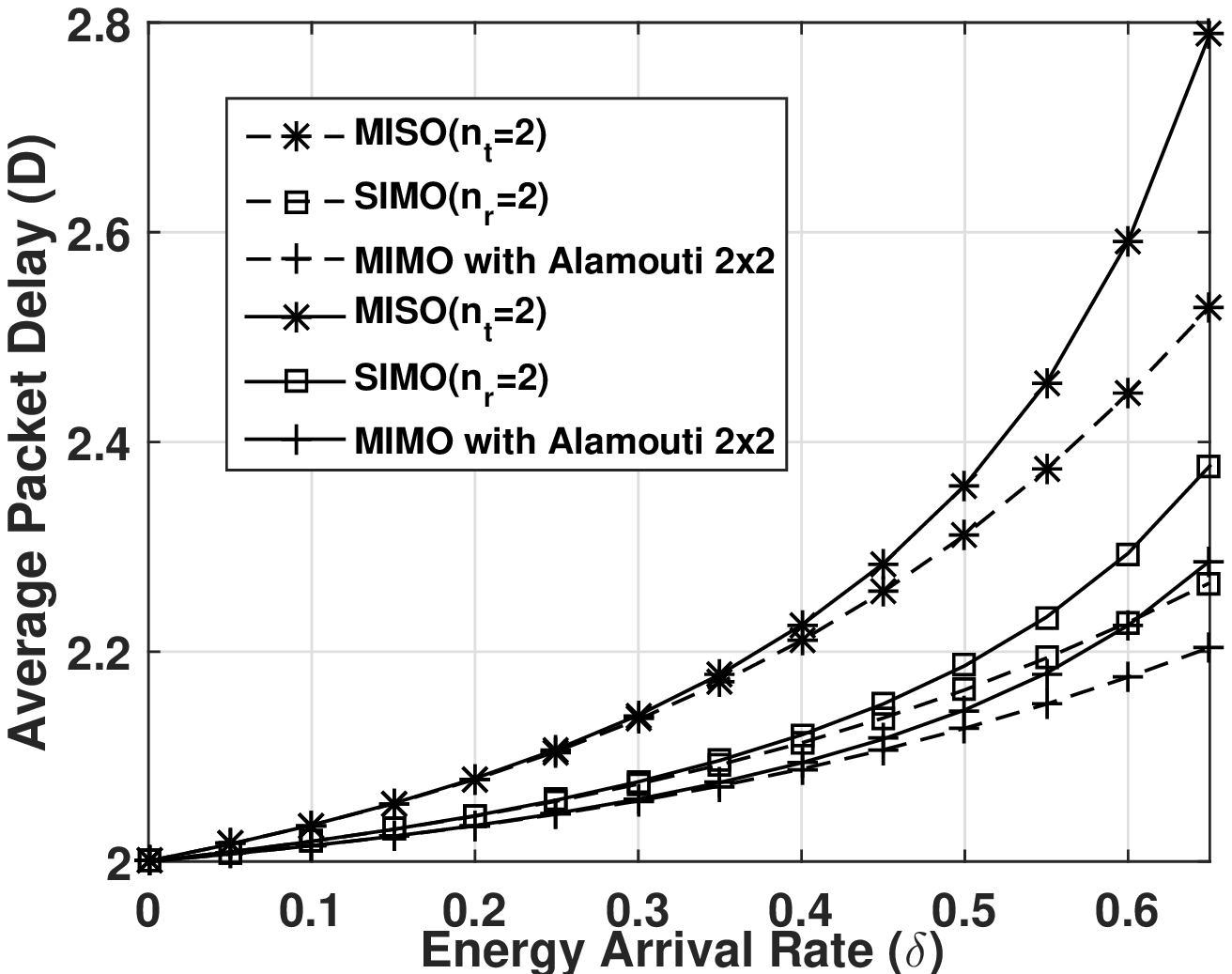}
	}
	\subfloat[\label{subfig-1:figure8}]{\includegraphics[trim={0.2cm 0cm 0.2cm .5cm},clip, height=2.4 in,width=3.3 in]{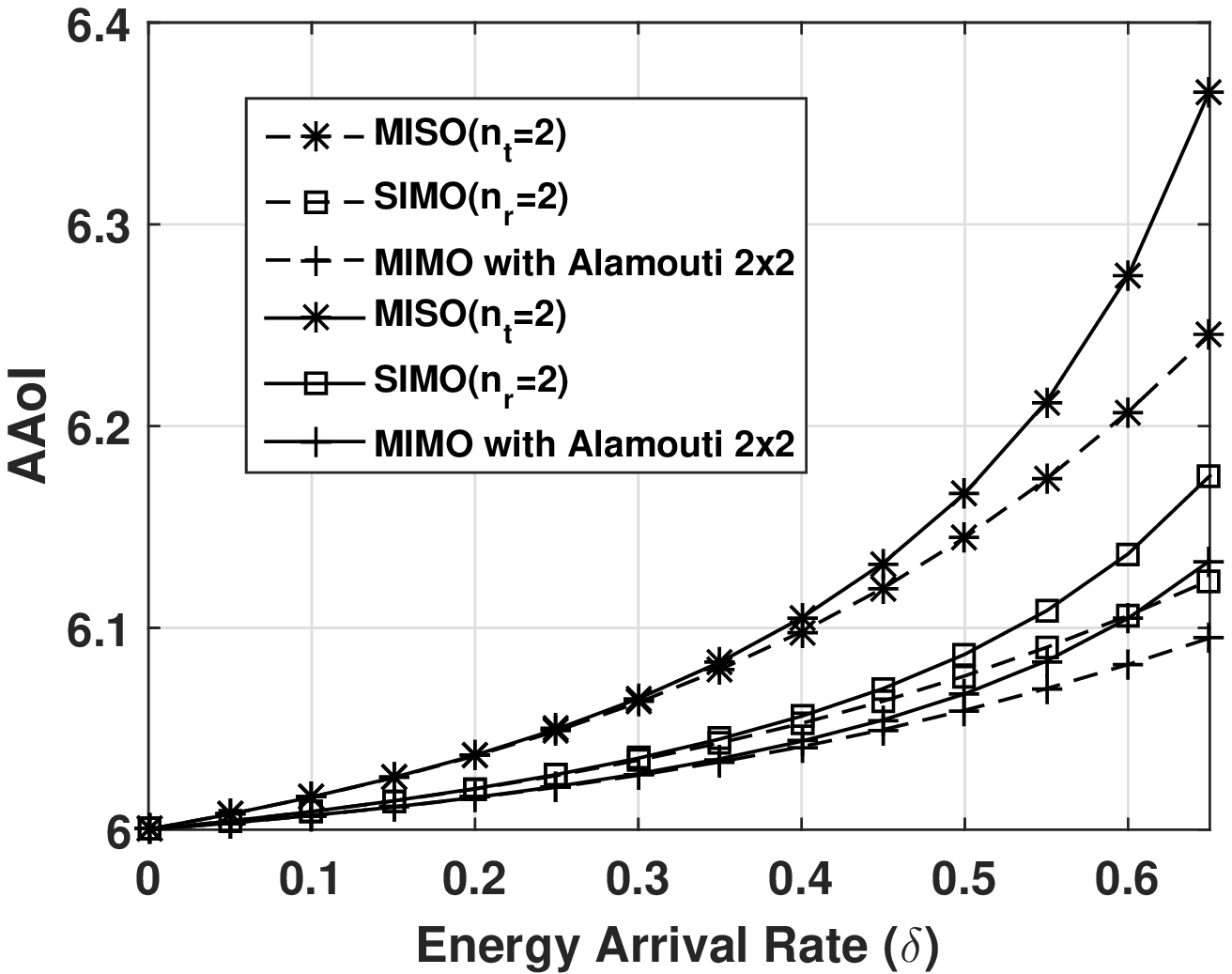}}
	\setlength{\abovecaptionskip}{-2pt} 
	\setlength{\belowcaptionskip}{-20pt} %
	\caption{(a) Average packet delay against energy arrival rate at jammer; (b) Average age of information against energy arrival rate at jammer:
		 $P=20$ dB, $P_J=20$ dB, $R=1$, $B=2$, $p_J=0.7$ and $\lambda=0.2$.}
\end{figure*}
%

\textit{In Figs.~\ref{fig:figure6}-\ref{subfig-1:figure16}: Dash and continuous curves correspond to battery with finite and infinite energy capacity, respectively.} In Fig.~\ref{fig:figure6}, the effect of energy arrival rate on the service rate is explored for different antenna configurations using the results in Theorems~\ref{th:secondtheorem} and \ref{th:thirdtheorem}. Note that MIMO with Alamouti provides the maximum service rate in comparison to SIMO and MISO systems with an increase in the number of antennas. In Figs. \ref{subfig-1:figure7} and \ref{subfig-1:figure8}, the average packet delay and AAoI are plotted against the energy arrival rate at the jammer for battery with finite and infinite energy capacity. It can be noticed that the AAoI and average packet delay is lowest for MIMO with Alamouti coding scheme in comparison to other cases. As the energy arrival rate $(\delta)$ increases at the jammer, the benefits of using Alamouti coding for the MIMO case are evident.
\begin{figure*}[htt]
	\centering
	\subfloat[\label{subfig-1:figure9}]{\includegraphics[trim={0.2cm 0cm 0.2cm .5cm},clip, height=2.4 in,width=3.3 in]{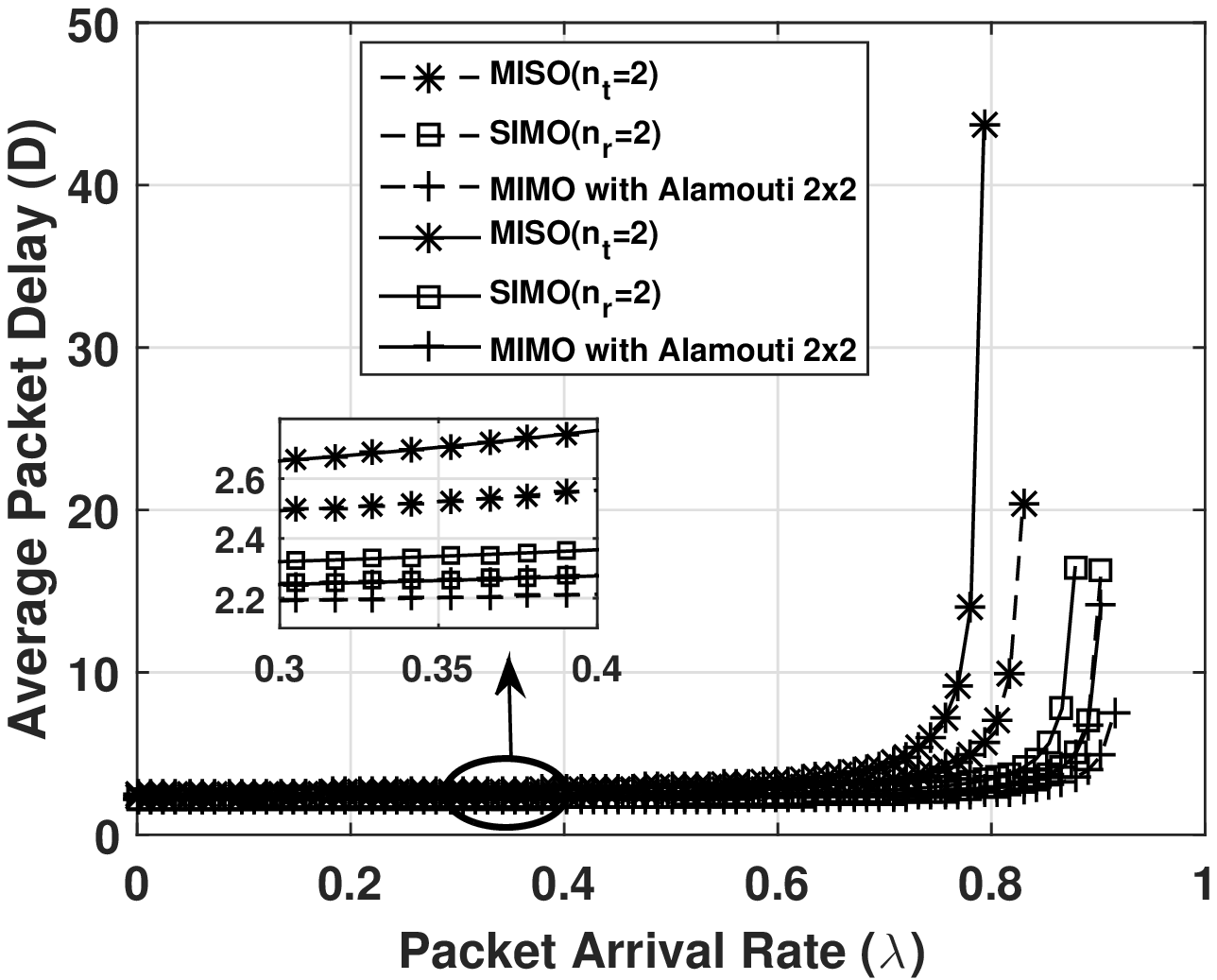}
	}
	\subfloat[\label{subfig-1:figure10}]{\includegraphics[trim={0.2cm 0cm 0.2cm .5cm},clip, height=2.4 in,width=3.3 in]{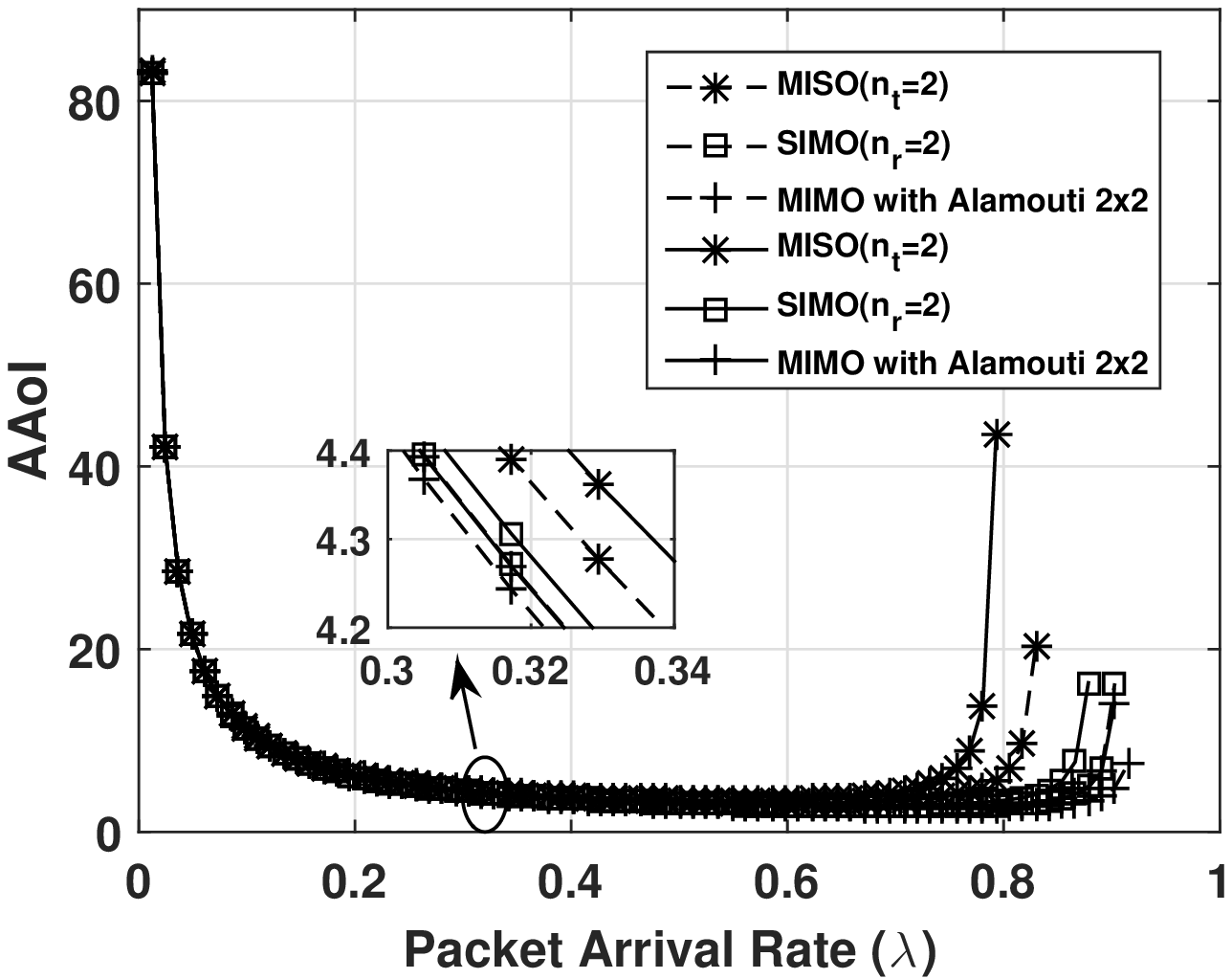}}
	\setlength{\belowcaptionskip}{-18pt} %
	\caption{(a) Average delay against packet arrival rate at transmitter; (b) Average age of information against packet arrival rate at transmitter:
		 $P=20$ dB, $P_J=20$ dB, $R=1$, $\delta=0.6$, $B=2$, and $p_J=0.7$.}
\end{figure*}

In Figs.~\ref{subfig-1:figure9} and \ref{subfig-1:figure10}, the average delay and AAoI are plotted against packet arrival rate at the transmitter for different antenna configurations, respectively. From Fig~\ref{subfig-1:figure10}, it can be observed that AAoI initially decreases with the increase in the arrival rate and then increases when the arrival rate goes beyond some specific value. The reason behind this observation is that when the $\lambda$ is low, the source does not often generate updates. Thus, the destination is not updated with new information often, and it results in high AAoI. When the arrival probability is high, there is an excessive queuing delay. Thus, the AAoI is high at the destination. This is because the received packets faced a queuing delay that affected the freshness of the information they carry. This is in contrast to the behavior of average packet delay with the increase in the arrival rate. \textit{This observation highlights the importance of AAoI as a measure to capture freshness since the performance metrics such as throughput and delay cannot capture this attribute.} It is worth noting that MIMO with Alamouti coding scheme can support the lowest value of average packet delay and AAoI with a higher value of the packet arrival rate $(\lambda)$. 
\begin{figure*}[htt]
	\centering
	\subfloat[\label{subfig-1:figure11}]{\includegraphics[trim={0.2cm 0cm 0.2cm .5cm},clip, height=2.4 in,width=3.3 in]{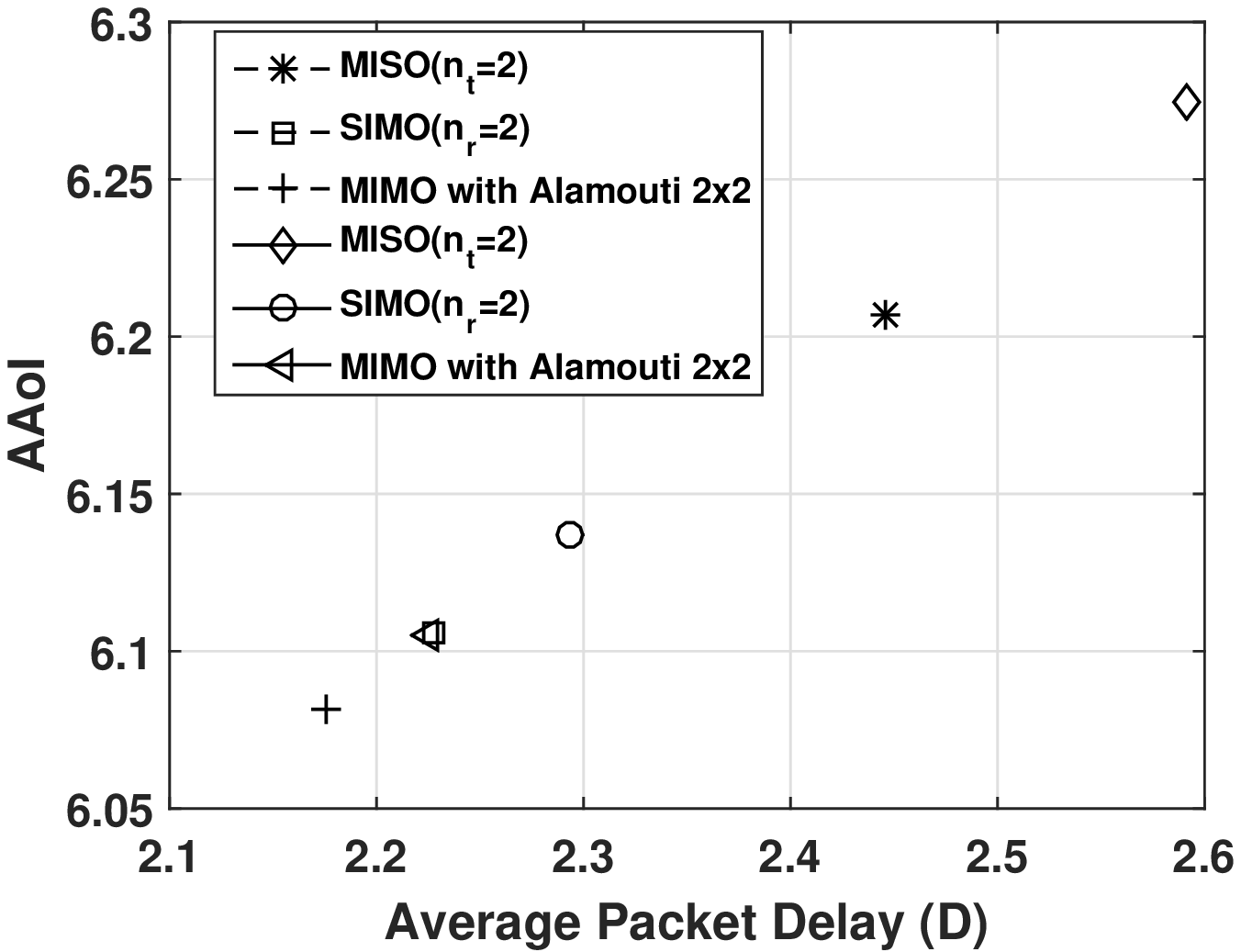}
	}
	\subfloat[\label{subfig-1:figure12}]{\includegraphics[trim={0.2cm 0cm 0.2cm .5cm},clip, height=2.4 in,width=3.3 in]{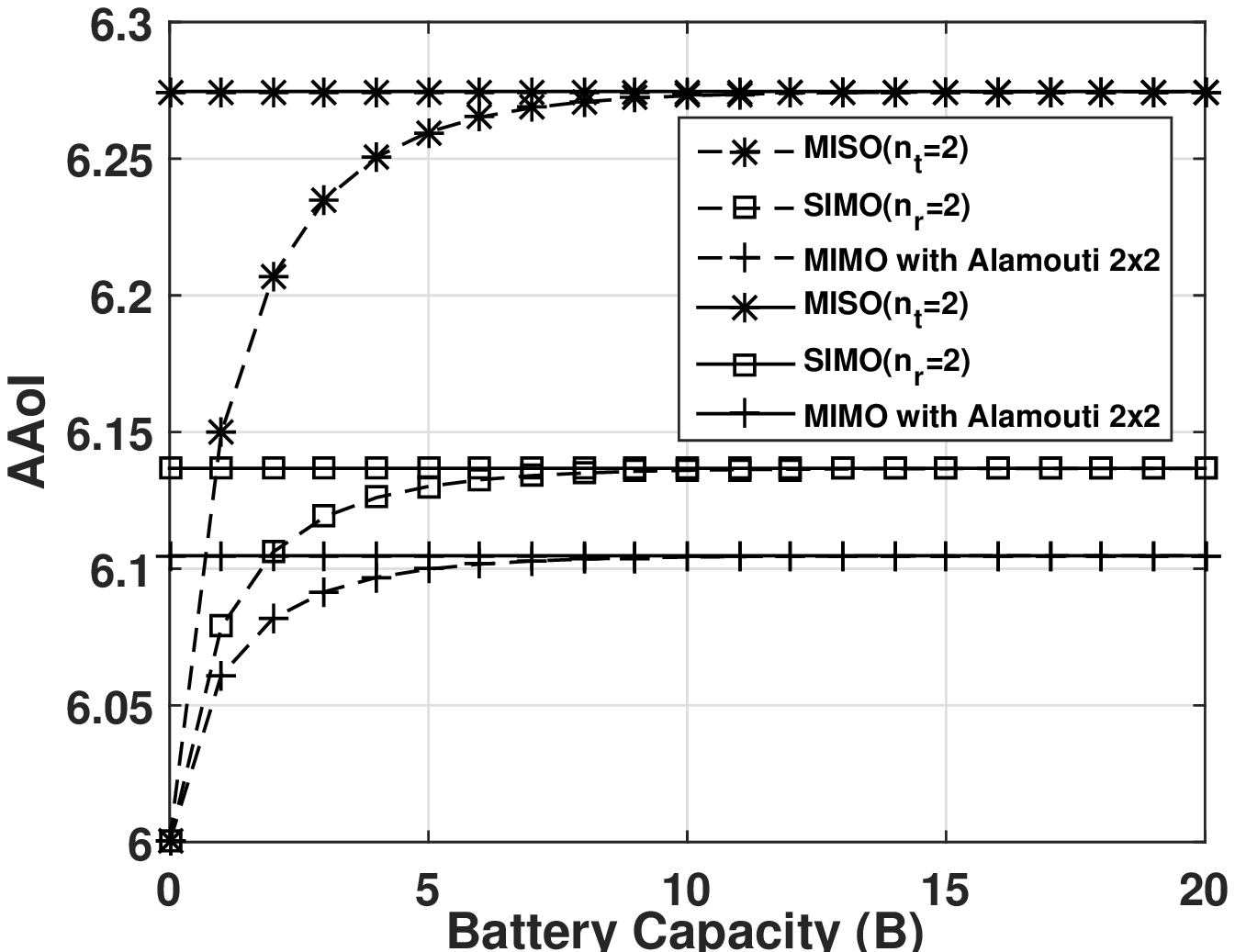}}
	\setlength{\belowcaptionskip}{-20pt} %
	\caption{(a) Average age of information against average packet delay ($B=2$); (b) Average age of information against energy harvesting battery capacity at jammer:
		$P=20$ dB, $P_J=20$ dB, $R=1$, $\delta=0.6$, $\lambda=0.2$, and $p_J=0.7$.}
\end{figure*}

Fig.~\ref{subfig-1:figure11} shows the interplay between AAoI and average packet delay for different antenna configurations. When the jammer has an energy battery of infinite capacity, the value of average packet delay and AAoI increases compared to the jammer with finite energy battery capacity. We can observe that MIMO with Alamouti guarantees a low value of average packet delay as well as AAoI. From Fig.~\ref{subfig-1:figure12}, one can see that when the battery size is around B = 10, the AAoI is similar to the case when the jammer has a battery of infinite capacity. Hence, the jammer cannot further increase the AAoI by increasing the battery size. 
\begin{figure*}[htt]
	\centering
	\subfloat[\label{subfig-1:figure13}]{\includegraphics[trim={0.2cm 0cm 0.2cm .5cm},clip, height=2.4 in,width=3.3 in]{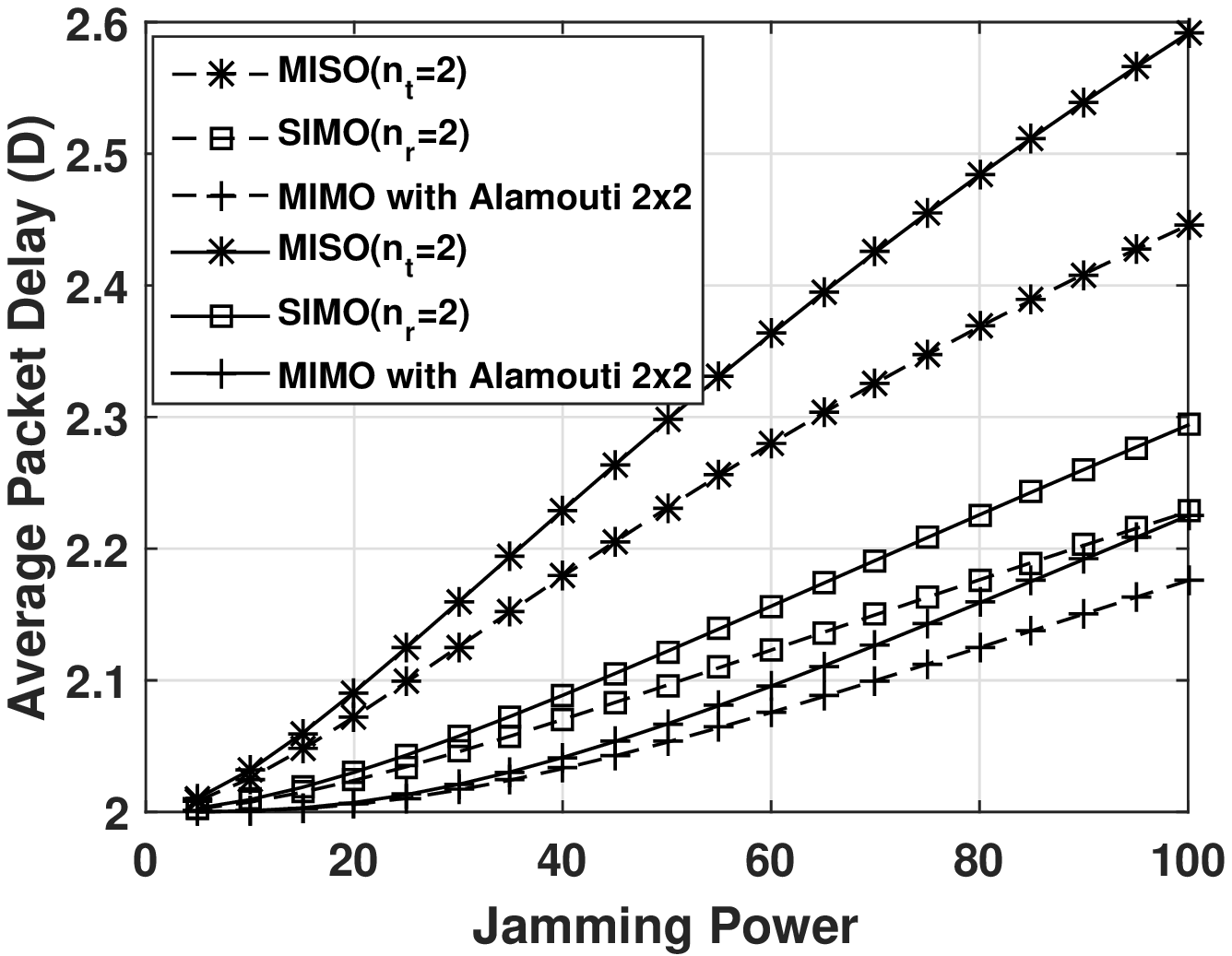}
	}
	\subfloat[\label{subfig-1:figure14}]{\includegraphics[trim={0.2cm 0cm 0.2cm .5cm},clip, height=2.4 in,width=3.3 in]{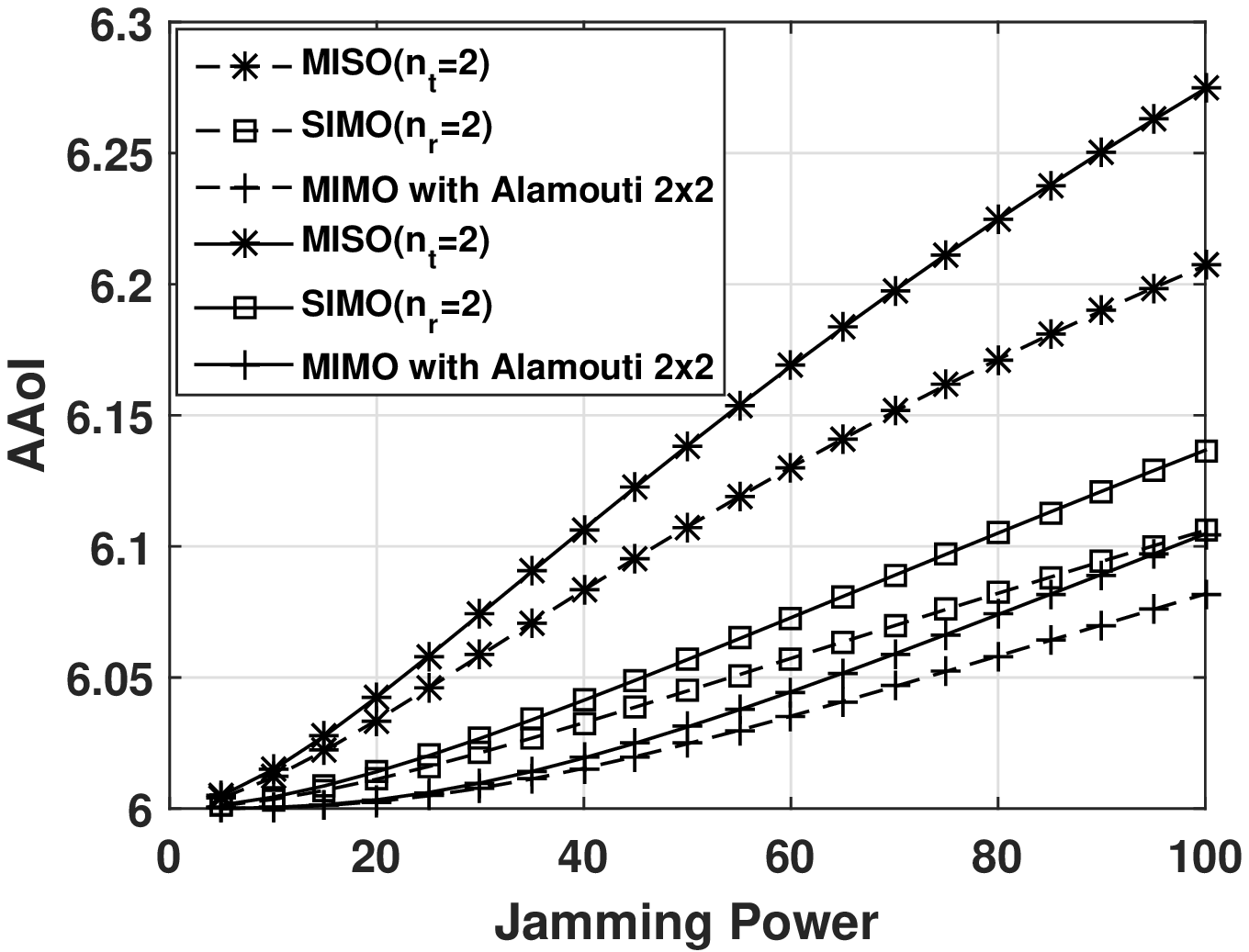}}\setlength{\abovecaptionskip}{-2pt}
	\setlength{\belowcaptionskip}{-20pt} %
	\caption{(a) Average packet delay against jamming power; (b) Average age of information against jamming power:
		$P=20$ dB, $R=1$, $\delta=0.6$, $\lambda=0.2$, $B=2$ and $p_J=0.7$.}
\end{figure*}
\begin{figure*}[htt]
	\centering
	\subfloat[\label{subfig-1:figure15}]{\includegraphics[trim={0.2cm 0cm 0.2cm .5cm},clip, height=2.4 in,width=3.3 in]{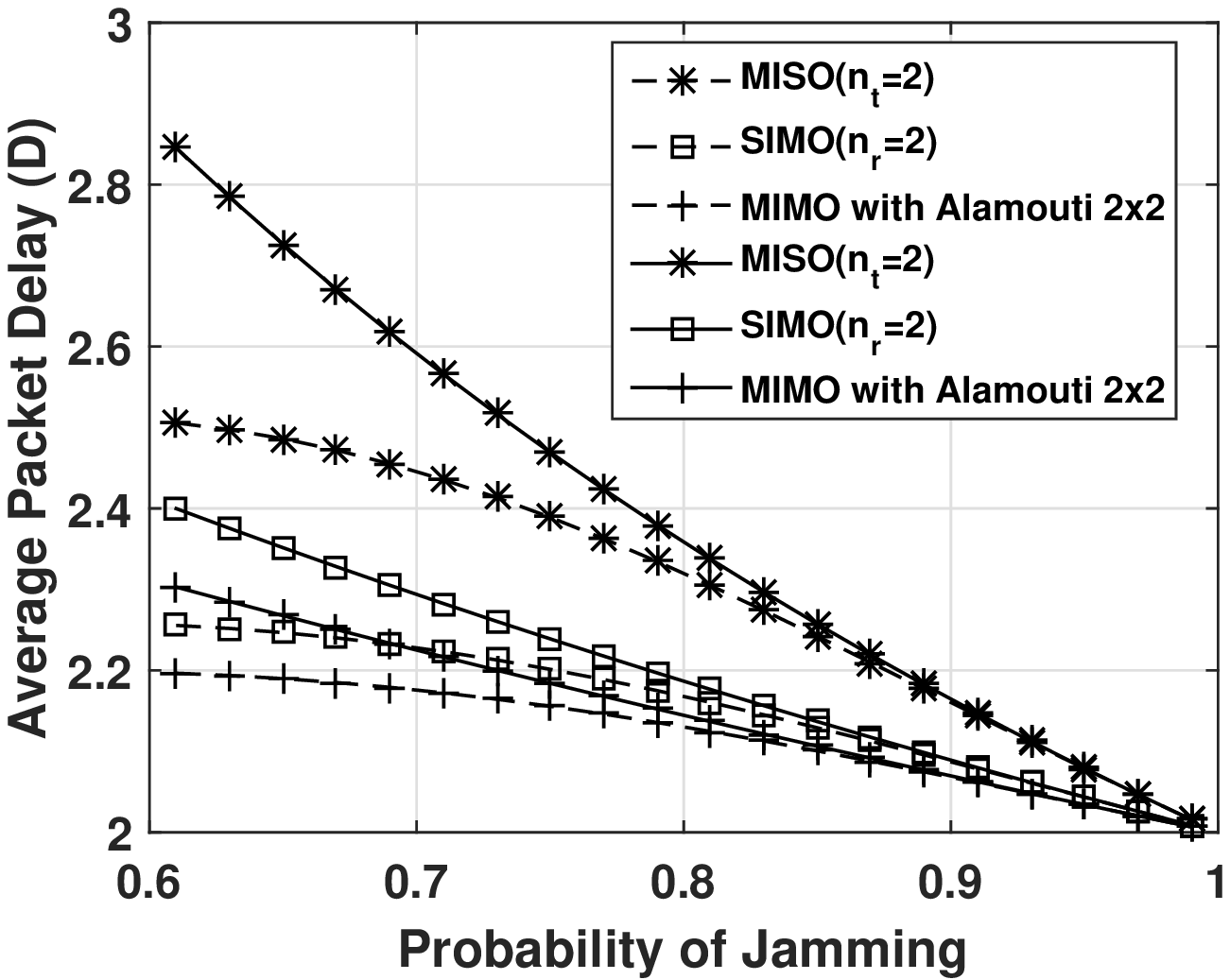}
	}
	\subfloat[\label{subfig-1:figure16}]{\includegraphics[trim={0.2cm 0cm 0.2cm .5cm},clip, height=2.4 in,width=3.3 in]{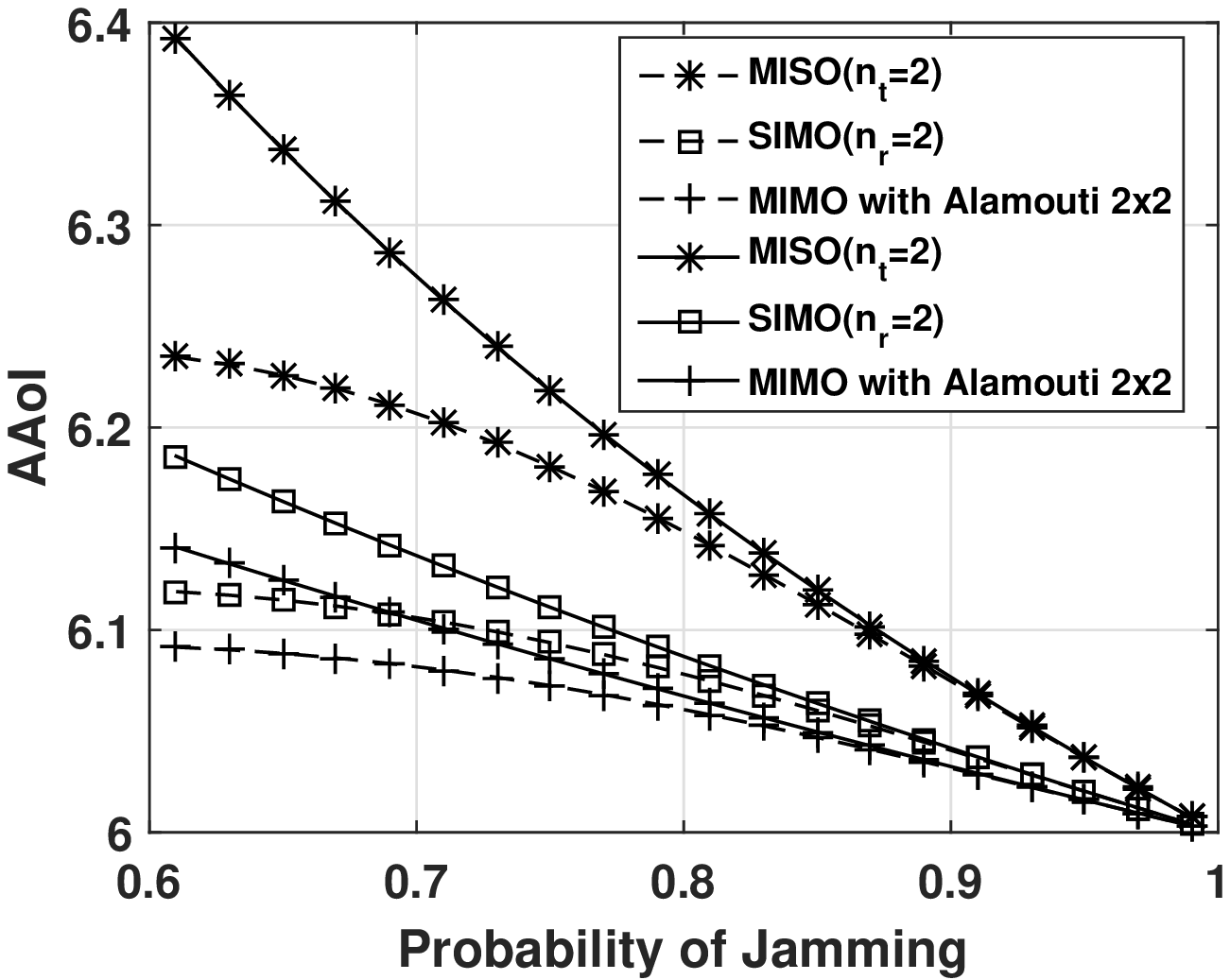}}
	\setlength{\abovecaptionskip}{-1pt}
	\setlength{\belowcaptionskip}{-15pt} %
	\caption{(a) Average packet delay against probability of jamming; (b) Average age of information against probability of jamming:
		$P=20$dB, $P_J=20$dB, $R=1$, $\delta=0.6$, $\lambda=0.2$, and $B=2$.}
\end{figure*}

Figs.~\ref{subfig-1:figure13} and \ref{subfig-1:figure14} show the impact of jamming power on average packet delay and AAoI. Figs.~\ref{subfig-1:figure15} and \ref{subfig-1:figure16} show the impact of the probability of jamming on average packet delay and AAoI.  From the plots, it can be observed that with increasing the jamming probability, the average packet delay and AAoI decrease as it is limited by the energy harvesting ability. It can be noticed that MIMO with Alamouti performs better in both cases, i.e., finite and infinite battery capacity, among all considered configurations. As the jamming power increases at the attacker, the benefits of space-time diversity are more prominent.
\begin{figure*}[htt]
	\centering
	\subfloat[Without delay constraint\label{subfig-1:figure17}]{\includegraphics[trim={0.2cm 0cm 0.2cm .5cm},clip, height=2.4 in,width=3.3 in]{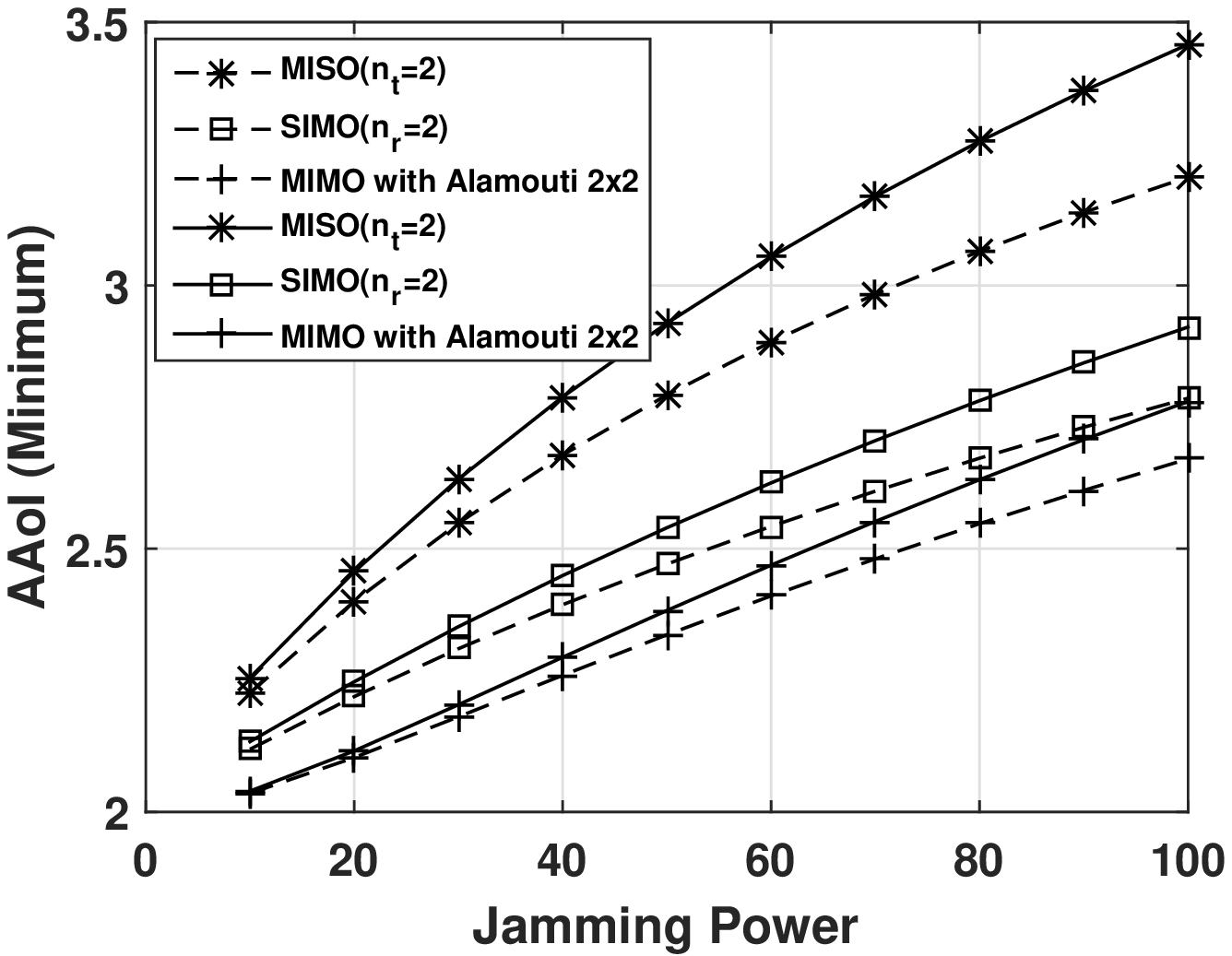}
	}
	\subfloat[With delay constraint $(D_{th} = 2.25)$\label{subfig-1:figure18}]{\includegraphics[trim={0.2cm 0cm 0.2cm .5cm},clip, height=2.4 in,width=3.3 in]{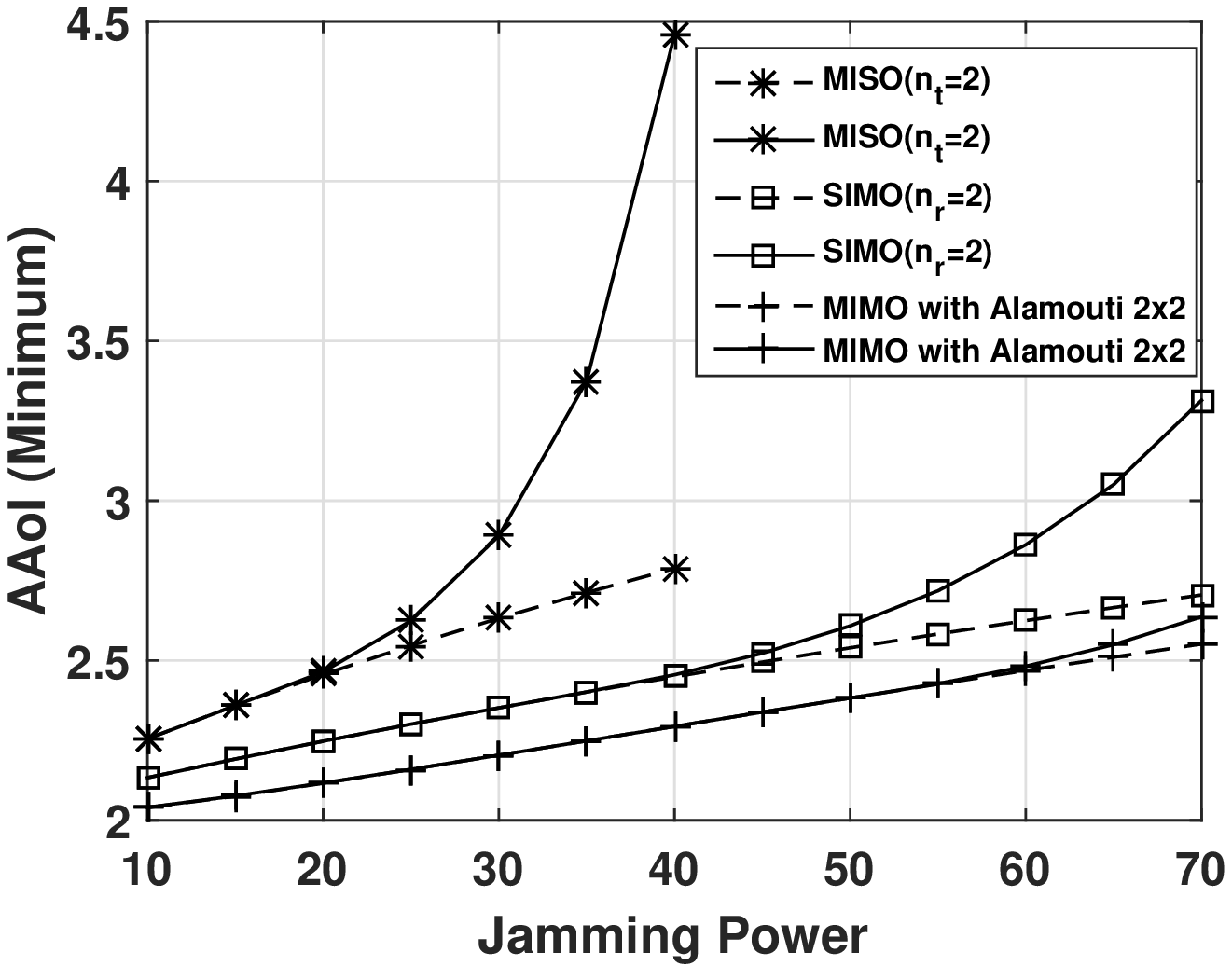}}
	\setlength{\belowcaptionskip}{-10pt} %
	\caption{Minimizing AAoI against jamming power: $P=20$ dB, $R=1$, $\delta=0.6$, $B=2$, and $p_J=0.7$. (a) Dashed curve: battery with finite energy capacity and continuous curves: battery with infinite energy capacity and (b) Dashed curve: without delay constraint and continuous curve: with delay constraint.}
\end{figure*}
Figs.~\ref{subfig-1:figure17} and \ref{subfig-1:figure18} show the minimized AAoI for optimal values of $\lambda$ against jamming power. In Fig.~\ref{subfig-1:figure17}, AAoI is plotted against the jamming power for the optimization problem stated in \eqref{AAoI_min_state}. In this case, there is no constraint on delay. From the figure, it can be observed that MIMO with Alamouti ensures minimum AAoI with an increase in the jamming power. In Fig.~\ref{subfig-1:figure18}, AAoI is plotted against the jamming power for the optimization problem stated in \eqref{AAoI_opt_2}. Recall that this optimization problem aims to minimize the AAoI with respect to arrival rate for a delay-sensitive system. It can be observed that the AAoI increases sharply with an increase in jamming power for the MISO system when there is a constraint on the delay. The Table~\ref{tab:Table_1} shows the corresponding value of $\lambda$ which minimizes the AAoI for MIMO $(2 \times 2)$ with Alamouti for both the optimization problems. However, in the case of MIMO with Alamouti coding, it can guarantee minimum AAoI among the different antenna configurations even if there is a delay constraint. It can be seen that with the increase in the jamming power, it is not possible to achieve low AAoI with stringent delay constraints. 
\begin{table}
	\centering
	\caption
	{MIMO $(2 \times 2)$ with Alamouti coding scheme : Optimum values of data arrival rate ($\lambda$) against jamming power for infinite energy battery capacity $(D_{th} = 2.25)$.}
	\begin{tabular}{|c|c|c|c|}
		\hline
		\multicolumn{2}{|c|}{Delay-Tolerant Traffic} & \multicolumn{2}{|c|}{Delay-Sensitive Traffic}\\ \hline
		\textbf{$P_J$} & \textbf{$\lambda$} & \textbf{$P_J$}   &\textbf{$\lambda$} \\ \hline
		10   & 0.9804  & 10   & 0.9800   \\ \hline
		15   & 0.9631  & 15   & 0.9626   \\ \hline
		20	 & 0.9439  & 20   & 0.9435   \\ \hline
		25	 & 0.9241  & 25   & 0.9240   \\ \hline
		30	 & 0.9044  & 30   & 0.9044   \\ \hline
		35	 & 0.8851  & 35   & 0.8852   \\ \hline
		40	 & 0.8665  & 40   & 0.8665   \\ \hline
		45	 & 0.8488  & 45   & 0.8487   \\ \hline
		50	 & 0.8319  & 50   & 0.8318   \\ \hline
		55	 & 0.8158  & 55   & 0.8045   \\ \hline
		60	 & 0.8006  & 60   & 0.7646   \\ \hline
		65	 & 0.7862  & 65   & 0.7217   \\ \hline
		70	 & 0.7726  & 70   & 0.6746   \\ \hline
	\end{tabular}\label{tab:Table_1}\vspace{-13pt}
\end{table}
\begin{figure}
	\centering
	\includegraphics[trim={0.2cm 0cm 0.2cm .5cm},clip, height=2.4 in,width=3.3 in]{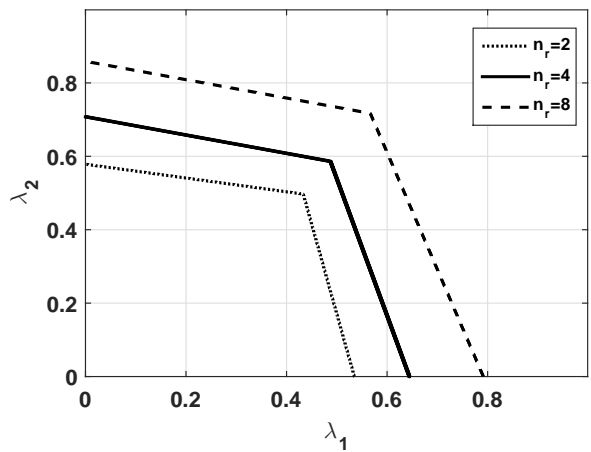}%
	\setlength{\belowcaptionskip}{-23pt}
	\caption{Stability region with $P_1 = P_2 = 10$dB, $ P_J = 20$dB , $\gamma_1 = 0.7$, $\gamma_2 = 0.5$, and $p_J=0.6$.}\label{fig:figure19}
\end{figure}
In Fig.~\ref{fig:figure19}, the stability region obtained from \eqref{rate_1} and \eqref{rate_2} for the $2$-user SIMO broadcast channel in the presence of a jammer is plotted. In this case, it is assumed that battery at the jammer has unlimited capacity and $\delta > p_J$. It can be observed that as the number of antennas at the receivers increases, the stability region also increases. 
\vspace{-8pt}
\section{Conclusions}
This work characterizes the performance of a point-to-point MIMO system under random arrival of data at the transmitter in the presence of a jammer with energy harvesting ability. The outage probability is derived for the considered system model for the Rayleigh fading scenario. The derived results on outage probabilities were used to characterize the \textit{average service rate}, \textit{average delay}, and \textit{AAoI} when the attacker has a battery with finite or unlimited capacity. The results illustrate the role of transmit, receive, space, and time diversity in mitigating jamming attacks and improving the performance of the system in the presence of a jammer. The developed results on average delay and AAoI show the relevance of these metrics under jamming attack when the transmitter has time-sensitive information and timely updates are of importance. The exploitation of time and space diversity helps to achieve superior performance even when the transmitter has time-sensitive data and less power budget. Moreover, the work also demonstrates how the developed results for the point-to-point case helps to determine the stability region for multi-user scenarios such as 2-user BC in the presence of jammer. Extending the results for the scenario when jammer has multiple antennas is an interesting direction for research. Another interesting direction of work is the characterization of stability region for multi-user scenario under jamming attack. 
\normalem
\bibliographystyle{IEEEtran}
\bibliography{mybiboo2} 
\end{document}